%% file: main.tex
\documentclass[acmsmall]{acmart}

\acmArticle{1}
\startPage{1}

\setcopyright{none}
\usepackage{algorithm}
\usepackage{amsmath}
\usepackage{amsthm}
\usepackage{amssymb}
\usepackage{stmaryrd}
\usepackage{url}
\usepackage{listings}
\usepackage{mdwlist}
\usepackage{pifont}
\usepackage{graphicx}
\usepackage{subcaption}
\usepackage{fancyvrb}
\usepackage{paralist}
\usepackage{caption}
\usepackage{comment}
\usepackage{bbm, dsfont}
\usepackage{mathrsfs}
\usepackage{xcolor}
\usepackage{tikz}
\usepackage{pgfplots}
\usepackage{esvect}
\usepackage{xspace}
\usepackage{array, multirow}
\usepackage{rotating, makecell}
\usepackage[noend]{algpseudocode}
\usepackage{balance}
\usepackage{wrapfig}

\newcommand{\irule}[2]%
   {\mkern-2mu\displaystyle\frac{#1}{\vphantom{,}#2}\mkern-2mu}

\newcommand{\system}{\textsc{Mediator}\xspace}

\newcommand{\termdelim}{~|~}
\newcommand{\ruledelim}{\quad}
\newcommand{\ins}{\texttt{ins}}
\newcommand{\del}{\texttt{del}}
\newcommand{\upd}{\texttt{upd}}

\renewcommand{\S}{{S}}

\newcommand{\M}{\vec{T}_U}
\newcommand{\Q}{\vec{T}_Q}
\newcommand{\awr}[2]{\langle {#1} \triangleleft {#2} \rangle}

\newcommand{\denot}[1]{{\llbracket #1 \rrbracket}}
\newcommand{\instance}{\Delta}
\newcommand{\lamvec}[2]{{\lambda \vec{#1}. \xspace #2} \xspace}
\newcommand{\tra}{\mathcal{T}_{\emph{RA}}}
\newcommand{\attnum}{\varsigma}
\newcommand{\universe}{\mathcal{P}}
\newcommand{\nil}{[ \ ]}
\newcommand{\hoare}[3]{\{{#1}\}~{#2}~\{{#3}\}}

\newcommand{\qbs}{\textsc{Qbs}\xspace}
\newcommand{\fiat}{\textsc{Fiat}\xspace}
\newcommand{\cosette}{\textsc{Cosette}\xspace}

\newcommand*{\comments}{} %comment to remove comments
\ifdefined\comments
\newcommand{\yuepeng}[1]{\textcolor{magenta}{\textbf{YUEPENG:} #1}}
\newcommand{\isil}[1]{\textcolor{red}{\textbf{I\c{S}IL:} #1}}
\newcommand{\todo}[1]{\textcolor{red}{#1}}
\else
\newcommand{\yuepeng}[1]{}
\newcommand{\isil}[1]{}
\newcommand{\todo}[1]{}
\fi

\newcommand{\emptydb}{\emptyset}

\begin{document}

\title{Verifying Equivalence of Database-Driven Applications}

\author{Yuepeng Wang}
\affiliation{
  \institution{University of Texas at Austin}
  \country{USA}
}
\email{ypwang@cs.utexas.edu}

\author{Isil Dillig}
\affiliation{
  \institution{University of Texas at Austin}
  \country{USA}
}
\email{isil@cs.utexas.edu}

\author{Shuvendu K. Lahiri}
\affiliation{
  \institution{Microsoft Research}
  \country{USA}
}
\email{shuvendu.lahiri@microsoft.com}

\author{William R. Cook}
\affiliation{
  \institution{University of Texas at Austin}
  \country{USA}
}
\email{wcook@cs.utexas.edu}

%% abstract must show up before \maketitle command
\input{abstract}



%% Keywords

\maketitle

\input{intro}

\input{overview}

\input{formalism}

\input{methodology}
\input{theory}

\input{automation}

\input{impl}

\input{eval}

\input{related}

\input{limitation}

\input{concl}

\input{ack}

\bibliography{main}

\newcommand*{\extended}{} % comment to remove appendices
\ifdefined\extended
\newpage
\appendix
\input{proof}

\fi

\end{document}

%% file: abstract.tex
\begin{abstract}
This paper addresses the problem of verifying equivalence between a pair of programs that operate over databases with different schemas. This problem is particularly important in the context of  web applications, which typically undergo database refactoring either for performance or maintainability reasons. While  web applications should have the same externally observable behavior before and after schema migration,  there are no existing tools for proving equivalence of such programs. This paper takes a first step towards solving this problem by formalizing the equivalence and refinement checking problems for database-driven applications.
We also propose a proof methodology based on the notion of \emph{bisimulation invariants} over \emph{relational algebra with updates} and describe a technique for synthesizing such bisimulation invariants. We have implemented the proposed technique in a tool called \system for  verifying equivalence between database-driven applications written in our intermediate language and evaluate our tool on 21 benchmarks extracted from textbooks and real-world web applications. Our results show that the proposed methodology can successfully verify 20 of these benchmarks.

\end{abstract}

%% file: intro.tex
\section{Introduction}\label{sec:intro}

Due to its usefulness in several application domains, the problem of checking equivalence between programs has been the subject of decades of  research in the programming languages community. For instance, in the context of translation validation, equivalence checking allows us to prove that the low-level code produced by an optimizing compiler is semantics-preserving with respect to the original program~\cite{tv-necula,tv-pnueli,tv-rinard,tv-zuck}. Other important applications of equivalence checking include \emph{regression verification}~\cite{regression1,regression2} and \emph{semantic differencing}~\cite{symdiff1,symdiff2,symdiff3}.

While there has been significant progress in the area of equivalence checking, existing techniques cannot be used to verify the equivalence of \emph{database-driven programs}. To see why verifying equivalence is important in this context, consider the scenario
in which a web application interacts with a relational database  to dynamically render a web page, and suppose that the database schema needs to be changed either for performance or maintainability reasons. 
In this case,  the developer  will need to migrate the database to the new schema and also re-implement the code that interacts with the database \emph{without} changing the observable behavior of the application. While this task of \emph{database refactoring} arises very frequently during the life cycle of web applications, it is also known to be quite hard and error-prone~\cite{ambler2007test,wikitech}, and several textbooks have been published on this topic~\cite{refactordb, refactorsql}.  As pointed out by Faroult and L'Hermite, \emph{``database applications are a difficult ground for refactoring''} because \emph{``small changes are not always what they appear to be''} and \emph{``testing the validity of a change may be difficult''}~\cite{refactorsql}.

Motivated by the prevalence of database-driven applications in the real world and the frequent need to perform schema migration, this paper proposes a new technique for verifying the equivalence of database-driven applications. Given a pair of  programs $P$, $P'$ that interact with two different databases, we would like to {automatically} construct a proof that $P$ and $P'$ are semantically equivalent.  Unfortunately, the problem of equivalence checking for database-driven applications introduces several challenges that are not addressed by previous work: First, it is unclear how to define equivalence in this context, particularly when the two database schemas are different.
Second, database-driven applications typically use declarative query languages, such as SQL, but, to the best of our knowledge, there are no automated reasoning tools for a rich enough fragment of SQL that captures realistic use cases.

In this paper, we formalize the equivalence checking problem for database-driven applications and propose a  verification algorithm for proving equivalence. Suppose that we are given two programs $P$, $P'$ that interact with databases $D$, $D'$. Let us also assume that each program comprises a set of database transactions (query or update) such that every transaction $t$ in $P$ has a corresponding transaction $t'$ in $P'$.
Our goal is to prove that  $P$ and $P'$ yield the same result on a pair of corresponding queries $q$, $q'$ whenever
we  execute the same sequence of update operations on $D$ and $D'$.

\begin{figure}[!t]
\centering
\includegraphics[scale=0.6]{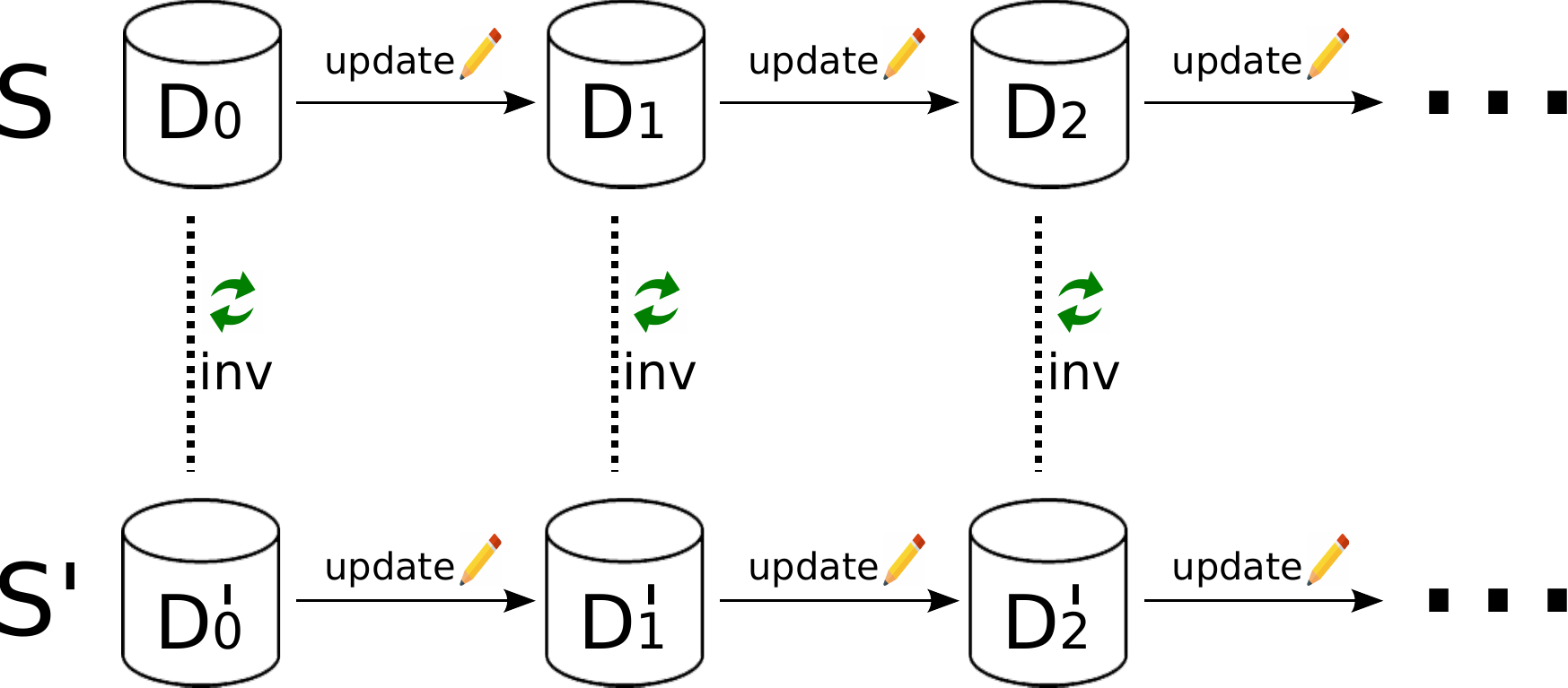}
\caption{Bisimulation invariant between two database-driven applications}
\label{fig:bisimilar}
\end{figure}

To prove equivalence between a pair of database-driven applications, our approach infers a so-called \emph{bisimulation invariant} that relates states of the two programs. In this context, program states correspond to database instances, so the bisimulation invariants relate a pair of database instances. As shown in Figure~\ref{fig:bisimilar},  our bisimulation invariants are preserved after each database transaction, and, in addition, they are strong enough to establish that any corresponding pair of queries must  yield the same result.

In the context of software verification, program invariants are typically expressed in some first-order theory supported by modern SMT solvers. Unfortunately, since the bisimulation invariants that we require in this context relate a pair of databases, they typically involve non-trivial relational algebra operators, such as join and selection.
To solve this difficulty, we consider the \emph{theory of relational algebra with updates}, $\tra$, and present an SMT-friendly encoding of $\tra $ into the theory of lists, which is supported by many SMT solvers.

Once we automate reasoning in relational algebra with updates,
a remaining challenge is to automatically infer  suitable bisimulation invariants. In this paper, we use the framework of monomial predicate abstraction to automatically synthesize conjunctive quantifier-free  invariants that relate two database states~\cite{mpa1,mpa2,mpa3}. Specifically, we identify a family $F$ of predicates that are useful for proving equivalence and  generate the strongest conjunctive bisimulation invariant over this universe. Towards this goal,we define   a strongest post-condition semantics of database transactions and automatically generate verification conditions whose validity establishes the correctness of a candidate bisimulation invariant.

We have implemented the proposed approach in a tool called \system for verifying equivalence between applications written in our intermediate representation (IR), which abstracts database-driven applications as a fixed set of queries and updates to the database. To evaluate our methodology,  we consider 21 database-driven applications translated into our IR and show that \system can successfully verify equivalence between benchmarks extracted from real-world web applications with up to hundreds of transactions. We also show that \system can handle challenging textbook examples that illustrate a wide spectrum of structural changes to the database schema. Overall, our experiments show that the proposed method is useful and practical:   \system can successfully verify the desired property {for 10 out of 11} real-world benchmarks in under 50 seconds on average.

To summarize, this paper makes the following key contributions:

\begin{itemize}
\item We introduce the equivalence and refinement checking problems for database-driven applications with different schemas and conduct an empirical study that demonstrates the practical relevance of this problem.
\item We present a sound and relatively complete proof methodology for showing equivalence of database-driven applications.
\item We show how to enable automated reasoning over relational algebra with updates ($\tra$) by presenting an SMT-friendly encoding of $\tra$ into the theory of lists.
\item We define a strongest postcondition semantics of database update operations and show how to automatically infer suitable simulation invariants for verifying equivalence and refinement.
\item We implement our approach in a tool called \system and experimentally evaluate it on 21 benchmarks taken from real-world applications and textbook examples. Our evaluation shows that \system can verify the desired property for {20 out of 21} benchmarks.
\end{itemize}

%% file: overview.tex
\section{Motivating Example} \label{sec:overview}

Consider a database-driven Connected Diagnostics Platform (\texttt{cdx}) \footnote{\url{https://github.com/instedd/cdx/commit/3006277ce4d1c8b097f1e7243ff205bf657ad3c0}} for notifying  patients about the results of their medical tests and alerting community health workers about (anonymous) positive test results in their area. This application interacts with a database that stores information about subscribers, laboratories, institutions and so on. In an earlier version of the application found on Github, the underlying database contains 16 relations and 125 attributes, but, at some point, the developers see a need to change the database schema and migrate it to a new format containing 17 relations and 131 attributes. Specifically,  the earlier version of the database contains a \texttt{Subscriber} relation with three attributes, namely \texttt{sid}, which corresponds to the subscriber id, \texttt{sname}, which is the name of the subscriber, and \texttt{filter}, which is used to filter out irrelevant subscriptions. In the updated version, the developers decide to refactor this information into two separate relations:

\vspace{0.1in}
\centerline{\texttt{Subscriber(sid, sname, fid\_fk)} \qquad 
\texttt{Filter(fid, fname, params)}
}
\vspace{0.1in}

The new database schema now contains an additional \texttt{Filter} relation, and the \texttt{fid\_fk} attribute in  \texttt{Subscriber}  is now a foreign key referring to the corresponding filter in the \texttt{Filter} relation. \footnote{The actual 
refactoring  from the Github commit history involves many other changes to the database schema, but we only consider this single modification to simplify the example.}

\input{fig-example}

After refactoring the database schema in this manner, the developers also re-implement the relevant parts of the code that interact with this database. In particular, Figure~\ref{fig:db-abs} shows the relevant functionality before and after the database refactoring. Both versions of the code contain three methods for updating the database, namely \texttt{createSub}, \texttt{deleteSub}, and \texttt{updateSub}, and two methods (\texttt{getSubName} and \texttt{getSubFilter}) for querying the database. However, the underlying implementation of these methods has changed due to the migration of the schema to a new format.  Nonetheless, we would like to be able to show that the application returns the same query results  before and after the schema migration.  This verification task is non-trivial because the database transactions in the two implementations are often structurally different and operate over different relations.

Let us now see how \system can be used to verify the equivalence of the two versions of the \texttt{cdx} application before and after  database refactoring. As mentioned in Section~\ref{sec:intro}, our method infers a \emph{bisimulation invariant} that relates two versions of the database. In the remainder of this discussion, let us use primed variables to refer to database relations and attributes in the  refactored database. For instance, \texttt{Subscriber'} refers to the version of the original \texttt{Subscriber} relation in the refactored database.

To come up with a suitable bisimulation invariant, \system first generates a finite universe of atomic predicates that \emph{could} be used to relate the two versions of the database. For example, one possible predicate is $\Pi_{sid, sname}(Subscriber) = \Pi_{sid', sname'}(Subscriber')$, which states that the \texttt{sid} and \texttt{sname} attributes in relation \texttt{Subscriber} correspond to \texttt{sid'} and \texttt{sname'} in \texttt{Subscriber'}, respectively. Given such  predicates, \system then tries to find a conjunctive formula that is provably a valid bisimulation invariant. 

As an example, let us consider the following candidate formula $\Phi$ in the  theory of relational algebra, which we will describe later in this paper:
\begin{gather*}
     \Pi_{sid, sname}(Subscriber) = \Pi_{sid', sname'}(Subscriber') \ \land \\
    \Pi_{sid, sname, filter}(Subscriber) = \Pi_{sid', sname', params'}(Subscriber' \Join Filter')
\end{gather*}
Essentially, this formula states that the {\tt sid} and {\tt sname} attributes of the {\tt Subscriber} relation are unchanged, and the  {\tt Subscriber} relation in the original database can be obtained by taking the natural join of {\tt Subscriber} and {\tt Filter} relations in the refactored database and then projecting the relevant attributes.

Using our verification methodology, we can automatically prove that the candidate formula $\Phi$ corresponds to a valid bisimulation invariant. In particular, assuming that $\Phi$ holds before each pair of  update operations $U, U'$ from the original and revised implementations, we can show that $\Phi$ still continues to hold after executing $U$ and $U'$.  In other words, this means that $\Phi$ is an \emph{inductive} bisimulation invariant. To prove the inductiveness of $\Phi$, we use a strongest postcondition semantics for database update operations as well an automated theorem prover for the theory of relational algebra with updates.

After finding an inductive bisimulation invariant $\Phi$, \system  still needs to ensure that $\Phi$ is  strong enough to prove equivalence. For this purpose, it considers every pair of queries $Q, Q'$ from the old and revised versions of the application and tries to prove that $Q$ and $Q'$ yield the same results.
Using the axioms of theory of relational algebra with updates, it can be shown that $\Phi \Rightarrow Q=Q'$ is logically valid for both of the queries \texttt{getSubName} and \texttt{getSubFilter} in this application.
Hence, \system is able to prove equivalence between these two programs even though they use databases that operate over different schemas.

%% file: fig-example.tex
\begin{figure}
\centering

\begin{subfigure}[b]{.49\linewidth}
\footnotesize
\begin{verbatim}

void createSub(int id, String name, String fltr)
  INSERT INTO Subscriber VALUES(id, name, fltr);

void deleteSub(int id)
  DELETE FROM Subscriber WHERE sid=id;

void updateSub(int id, String name, String fltr)
  UPDATE Subscriber SET filter=fltr WHERE sid=id;
  UPDATE Subscriber SET sname=name WHERE sid=id;

List<Tuple> getSubName(int id)
  SELECT sname FROM Subscriber WHERE sid=id;

List<Tuple> getSubFilter(int id)
  SELECT filter FROM Subscriber WHERE sid=id;

\end{verbatim}

\caption{Before Refactoring}
\end{subfigure}
\vrule~
\begin{subfigure}[b]{.49\linewidth}    
\footnotesize
\begin{verbatim}
void createSub(int id, String name, String fltr)
  INSERT INTO Subscriber VALUES(id, name, UUID_x);
  INSERT INTO Filter VALUES(UUID_x,
         "Filter for " + name + " subscriber", fltr);

void deleteSub(int id)
  DELETE FROM Filter WHERE fid IN
    (SELECT fid_fk FROM Subscriber WHERE sid=id);
  DELETE FROM Subscriber WHERE sid=id;

void updateSub(int id, String name, String fltr)
  UPDATE Filter SET params=fltr WHERE fid IN
    (SELECT fid_fk FROM Subscriber WHERE sid=id);
  UPDATE Subscriber SET sname=name WHERE sid=id;

List<Tuple> getSubName(int id)
  SELECT sname FROM Subscriber WHERE sid=id;

List<Tuple> getSubFilter(int id)
  SELECT params FROM Filter JOIN Subscriber
    ON fid=fid_fk WHERE sid=id;
\end{verbatim}

\caption{After Refactoring}
\end{subfigure}

\caption{Sample Database Refactoring. \texttt{UUID\_x} is a unique \texttt{fid} of \texttt{Filter} relation.}
\label{fig:db-abs}
\end{figure}

%% file: formalism.tex
\section{Problem Statement} \label{sec:formalism}

In this section, we formalize the syntax and semantics of database-driven applications and precisely define the {equivalence and refinement checking} problems in this context.

\subsection{Language Syntax}

\input{fig-ir}

In the remainder of this paper, we represent a database-driven application $P$ as a tuple $(\S, \M, \Q)$, where $\S$ is the schema of the underlying database, $\M$ is a vector of database update transactions, and $\Q$ is a vector of database queries (see Figure~\ref{fig:ir}). We collectively refer to any update or query in $\M \cup \Q$ as a \emph{database transaction} and  denote the $i$'th transaction in $\M$ (resp. in $\Q$) as $U_i$ (resp. $Q_i$). Let us now take a closer look at the syntax of the language from Figure~\ref{fig:ir}.

\vspace{0.1in}\noindent
\emph{Database schema.} The database schema $\S$ provides a logical view of how the database organizes its data. In particular, the schema describes all relations (i.e., tables) stored in  the database as well as the typed attributes for each relation. More precisely, we represent the schema $\S$ as a mapping from table names $R$ to their corresponding record types  $\{a_1: \tau_1; \ldots; a_n: \tau_n \}$, which indicates that attribute $a_i$ of table $R$ has type $\tau_i$. In the rest of the paper, we use the notation $\emph{dom}(\S)$ to denote the set of tables stored in the database.

\vspace{0.1in}\noindent
\emph{Update transactions.}
An update transaction $\lambda \vec{v}. U  \in \M$ contains a {sequence} of database update operations, including insertion, deletion, and modification. Specifically, the language construct  $\ins(R, t)$ models the  insertion of tuple $t$ into relation $R$, where $R \in \emph{dom}(\S)$ and tuple $t$ is represented as a mapping from attributes to symbols (variable or constant). Similarly, the statement $\del(R, \phi)$ removes all tuples satisfying predicate $\phi$ from $R$, and $\upd(R, \phi, a, v)$  assigns value $v$ to the $a$ attribute of all tuples satisfying predicate $\phi$ in $R$. We assume that each database transaction occurs atomically (i.e., either all or none of the updates are committed).

\vspace{0.1in}\noindent
\emph{Query transactions.}
In our language,  query transactions $\lambda \vec{v}. Q \in \Q$ are expressed as  relational algebra expressions involving projection ($\Pi$), selection ($\sigma$), join ($\Join$), union ($\cup$), and difference ($-$) operators. While our language allows general theta joins of the form $R_1 \Join_\phi R_2$, we abbreviate natural joins using the notation $R_1 \Join R_2$.

\subsection{Language Semantics}\label{sec:semantics}

To define the formal semantics of database-driven applications, we first need to define what we mean by an \emph{input} to the programs defined in Figure~\ref{fig:ir}. Since we consider a model in which the user interacts with the application by performing a sequence of updates and queries to the database, we consider a program input to be an \emph{invocation sequence} $\omega$ of the form:
\[
    \omega = (i_1, \sigma_1); \ldots; (i_{n-1}, \sigma_{n-1}); (i_{n}, \sigma_{n})
\]
where each $i_j$ specifies an update transaction $\lambda \vec{v}. U_{i_j}$ for $j \in [1, n)$ and $\sigma_j$ is its corresponding valuation, mapping values of formal parameters $\vec{v}$ to their concrete values. The last element  $(i_{n}, \sigma_{n})$ in the invocation sequence always corresponds to a \emph{query transaction} with corresponding valuation~$\sigma_n$.

\input{fig-ir-ex}

\begin{example}
Consider the motivating example from Figure~\ref{fig:db-abs}. This program can be expressed in our intermediate language as shown in Figure~\ref{fig:ir-ex}.  Now, consider the following  invocation sequence, assuming that the five transactions are indexed $1-5$ from top to bottom:
\[
    \omega = (1, [\emph{id} \mapsto 100, \emph{name} \mapsto Alice, \emph{fltr} \mapsto \emph{Filter1}]); (4, [\emph{id} \mapsto 100])
\]
This sequence indicates that the user first invokes the update transaction \texttt{createSub(100, Alice, Filter1)}, followed by the query transaction \texttt{getSubName(100)}.
\end{example}

\vspace{0.1in} \noindent
{\bf \emph{Remark.}} The reader may wonder why our definition of \emph{program input} requires all elements in an invocation sequence to be update transactions, except for the last one which  is always a  query.  Our motivation here is to simplify the formalization of equivalence.  Recall that we consider two programs to be equivalent if they yield the same answers to their corresponding queries under the same sequence of updates to the database. Since we can model invocation sequences that contain multiple queries as two different inputs, it suffices to consider inputs that contain a single query. Furthermore,  we can disregard invocation sequences that do not contain any queries because the program does not ``return" any output on such an input.
\vspace{0.1in}

\input{fig-semantics}

Figure~\ref{fig:semantics} defines the denotational semantics for the language presented in Figure~\ref{fig:ir}. We use the notation $\denot{P}_\omega$ to represent the result of the last query in  $\omega$ after performing all updates  on an \emph{empty database}. Since any {reachable} database state can be modeled using a suitable sequence of insertions to an empty database,  this assumption does not result in a loss of generality.

In our semantics, we model \emph{database instances} $\instance$ as a mapping from relation names to a list of tuples.~\footnote{We model relations as lists rather than bags because many libraries provide database interfaces based on ordered data structures.} Similarly, we  model tuples as a mapping from attribute names to their corresponding values. Given a program $P$, an invocation sequence $\omega = \omega'; (n, \sigma)$, and a database instance $\instance$, we first obtain a new instance $\instance'$ by  running $\omega'$ on $\instance$ and then evaluate the query $Q_n$ on database instance $\instance'$ and valuation $\sigma$. Observe that the result of a program is represented  as a list of lists rather than as relations (list of maps). That is, our semantics disregards the names of attributes to enable meaningful comparison between programs over different schemas.

The semantics for update transactions in Figure~\ref{fig:semantics} are described using the familiar list combinators, such as \emph{append}, \emph{filter}, \emph{map}, and \emph{fold}. In particular, $\denot{U}_{ \sigma, \instance}$ yields the  database instance after executing update operation $U$ with input $\sigma$ on database $\instance$. For example,  consider the semantics for $\ins(R, t)$: To obtain the new database instance, we first evaluate $t$ under valuation $\sigma$, where the notation $t[\sigma]$ denotes applying substitution $\sigma$ to term $t$. The entry for relation $R$ in the new database instance is obtained by appending the tuple $t[\sigma]$ to table $\instance(R)$, which is represented as a list of tuples. Similarly, $\del(R, \phi)$ filters from list $R$ the set of all tuples that do not satisfy predicate $\phi$. Finally, $\upd(R, \phi, a, v)$ first obtains a new relation $R_1$ that contains all tuples in $R$ that do not satisfy $\phi$.  It then also filters out all tuples of $R$ that satisfy predicate $\phi$, and updates the $a$ attribute of each such tuple to a new value $v[\sigma]$. The new entry for $R$ is then obtained by concatenating these two lists.

Let us now turn our attention to the semantics of query transactions (third part of Figure~\ref{fig:semantics}). The semantics are defined inductively, with the first rule for $R$ being the base case. Since $R$ corresponds to the name of a database table, we obtain the query result by simply looking up  $R$ in  $\instance$. As another example,  consider the semantics of the selection ($\sigma$) operator. Given a query of the form $\sigma_{\phi}(Q)$, we first recursively evaluate $Q$ and obtain the query result $T = \denot{Q}_{\sigma, \instance}$.
We then evaluate the predicate $\phi$ under $\sigma$ and obtain a \emph{symbolic} predicate $p = \denot{\phi}_{ \sigma, \instance, x}$. In particular, predicate $p$ is symbolic in the sense that it refers to a variable $x$, which ranges over tuples in $T$. We obtain the final query result by filtering out those rows of $T$ that do not satisfy predicate~$\lambda x. p$.

The final part of Figure~\ref{fig:semantics} describes predicate semantics inductively. For instance,  consider a predicate of the form $a \odot v$, where $a$ is an attribute, $\odot$ is a (logical) binary operator, and $v$ is a symbol (variable or constant). {Since the predicate  takes as input a tuple $x$}, we evaluate attribute $a$ by looking up $a$ in $x$, which is represented as a mapping from attributes to values.  Thus, the evaluation of the predicate is given by $\emph{lookup}(x, a) \odot v[\sigma]$, where the notation $v[\sigma]$ applies substitution $\sigma$ to symbol $v$.

\subsection{Equivalence and Refinement}
Having defined the semantics of database-driven applications, we are now ready to precisely state
the notion of semantic equivalence in this context:

\begin{definition}\textbf{(Program equivalence)}\label{def:equiv}
A database-driven program $P'$ is said to be semantically equivalent to another program $P$, denoted $P' \simeq P$, if and only if  executing $\omega$ on $P'$ yields the same  result as executing $\omega$ on $P$ for any invocation sequence $\omega$, i.e.,
\[
    P' \simeq P \triangleq \forall \omega . ~~ \llbracket P' \rrbracket_{\omega} = \llbracket P \rrbracket_\omega
\]
\end{definition}

In the above definition, we assume that $P$ and $P'$ have the same number of query and update transactions.~\footnote{Note that this assumption is not a fundamental limitation. If the number of transactions in $P$ and $P'$ are different, we can  either allow users to specify the correspondence between transactions in $P$ and $P'$, or alternatively, we can enumerate all possible correspondences and see if we can prove equivalence under some (possibly one-to-many or many-to-one) mapping between transactions in $P$ and $P'$.
} If this condition does not hold, we can immediately conclude that $P$ and $P'$ are not equivalent because some inputs that are valid for $P$ are not valid for $P'$ or vice versa. We also assume that $P$ and $P'$ have   transactions that are supposed to be functionally equivalent at the same index; otherwise, the transactions can be syntactically re-arranged. Under these assumptions, our definition effectively states that programs $P$ and $P'$ are equivalent whenever their corresponding queries yield the same result under the same sequence of update transactions to the database.

While there are many real-world scenarios in which we would like to prove equivalence, there are also some cases where one application \emph{refines} the other. For instance, consider a situation in which a web application developer  changes the database schema for performance reasons, but also decides to add some new piece of information to the underlying database such that query results also include this new information. In this scenario, the updated version of the application will  \emph{not} be semantically equivalent to its prior version, but we would still like to verify that adding new features does not break existing functionality. Towards this goal, we also formally define what it means for a database-driven program $P'$ to \emph{refine} another program $P$.

\begin{definition}\textbf{(Valuation refinement)}
Consider two valuations $\sigma$ and $\sigma'$. We say that $\sigma'$ is a refinement of $\sigma$, denoted $\sigma' \preceq \sigma$, if and only if $\sigma'$ maps the variables that occur in $\sigma$ to the same values as  in $\sigma$. In other words,
\[
    \sigma' \preceq \sigma \triangleq \forall x \in \emph{dom}(\sigma), ~~ \sigma'(x) = \sigma(x)
\]
\end{definition}

\begin{definition}\textbf{(Input refinement)}
Given invocation sequences $\omega = (i_1, \sigma_1) (i_2, \sigma_2) \ldots (i_n, \sigma_n)$ and $\omega' = (i'_1, \sigma'_1) (i'_2, \sigma'_2) \ldots (i'_n, \sigma'_n)$, we say that $\omega'$ refines $\omega$, denoted  $\omega' \preceq \omega$, if and only if $\omega'$ has the same index sequence as $\omega$ and the valuations in $\omega'$ refine the corresponding valuations in $\omega$. That is,
\[
    \omega' \preceq \omega \triangleq \forall k \in [1, n]. ~~ i_k = i'_k \land \sigma'_k \preceq \sigma_k
\]
\end{definition}

Using these  definitions, we can now also state what it means for an application to refine another one:

\begin{definition}\textbf{(Program refinement)}\label{def:ref}
Program $P'$ is said to refine another program $P$, denoted  $P' \preceq P$, if and only if,
for any invocation sequences $\omega', \omega$ satisfying $\omega' \preceq \omega$, executing $\omega$ on $P$ yields a relation that is a projection of executing $\omega'$ on $P'$, i.e.,
\[
    P' \preceq P \triangleq \forall \omega, \omega'. ~~ \omega' \preceq \omega \rightarrow \exists L. ~  \llbracket \Pi_L(P') \rrbracket_{\omega'} = \llbracket P \rrbracket_\omega
\]
\end{definition}
\noindent where $\Pi_L((\S, \M, \Q)) = (\S, \M, \Q')$ such that $Q'_i = \Pi_L(Q_i)$ for all $Q_i$ in $\Q$.

As in Definition~\ref{def:equiv}, we require that $\omega'$ and $\omega$ are valid inputs for $P'$ and $P$ respectively. However, unlike in Definition~\ref{def:equiv}, we do not assume that $P'$ contains the same number of update and query transactions in $P$. Our definition simply disregards the new transactions that are added by $P'$ and only considers invocation sequences that are valid for both. Thus, intuitively, if an application $P'$ refines $P$, the query results of $P$ can be obtained by applying a projection to the corresponding query results in $P'$.

%% file: fig-ir.tex
\begin{figure}
\centering
\[
\begin{array}{l l c l}
    \emph{Program} & P &:=& (\S, \M, \Q) \\
    \emph{Schema} & \S & := & R \rightarrow \{ a_1: \tau_1; \ldots; a_n: \tau_n \} \\
    \emph{Transaction} & T & := & \lambda \vec{v}.\ U \termdelim \lambda \vec{v}.\ Q \\
    \emph{Update} & U &:=& \ins(R, \{a_1: v_1, \ldots, a_n: v_n \}) \termdelim \del(R, \phi) \termdelim \upd(R, \phi, a, v) \termdelim U;U \\
    \emph{Query} & Q &:=& R \termdelim \Pi_{\psi}(Q) \termdelim \sigma_\phi(Q) \termdelim Q \Join_\phi Q \termdelim Q \cup Q \termdelim Q - Q \\
    \emph{Attribute list} & \psi &:=& a \termdelim \psi, \psi \\
    \emph{Predicate} & \phi &:=& p \termdelim \phi \land \phi \termdelim \phi \lor \phi \termdelim \neg \phi \\
    \emph{Atomic predicate} & p &:=& a \odot a \termdelim a \odot v \termdelim a \in Q \\
    \emph{Operator} & \odot &:=& \leq \termdelim < \termdelim = \termdelim \neq \termdelim > \termdelim \geq \\
\end{array}
\]
\vspace{-0.05in}
\[
\begin{array}{r c l c r c l}
 \tau & \in & \{ Int, String, \ldots \} & & R & \in & \emph{Relation} \\
    a & \in & \emph{Attribute} & & v & \in & \emph{Variable} \cup \emph{Constant}
\end{array}
\]
\caption{Grammar of the language used in our formalization}
\label{fig:ir}
\end{figure}

%% file: fig-ir-ex.tex
\begin{figure}
\centering

\begin{subfigure}[b]{.49\linewidth}
\footnotesize
\begin{verbatim}

void createSub(int id, String name, String fltr)
  ins(Subscriber, (id, name, fltr))

void deleteSub(int id)
  del(Subscriber, sid=id)

void updateSub(int id, String name, String fltr)
  upd(Subscriber, sid=id, filter, fltr)
  upd(Subscriber, sid=id, sname, name)

List<Tuple> getSubName(int id)
\end{verbatim}
$\ \ \ \ \Pi_{\tt sname}(\sigma_{\tt sid=id}(\tt Subscriber))$

\begin{verbatim}

List<Tuple> getSubFilter(int id)
\end{verbatim}
$\ \ \ \ \Pi_{\tt filter}(\sigma_{\tt sid=id}(\tt Subscriber))$

\vspace{0.1in}
\caption{Before Refactoring}
\end{subfigure}
\vrule~\ 
\begin{subfigure}[b]{.49\linewidth}    
\footnotesize
\begin{verbatim}
void createSub(int id, String name, String fltr)
  ins(Filter', (UUID_x, name, fltr))
  ins(Subscriber', (id, name, UUID_x))

void deleteSub(int id)
\end{verbatim}
$\ \ \ \ \del(\texttt{Filter', fid'} \in \Pi_{\tt fid\_fk'}(\sigma_{\tt sid'=id}(\tt Subscriber')))$
\begin{verbatim}
  del(Subscriber', sid'=id)

void updateSub(int id, String name, String fltr)
\end{verbatim}
$\ \ \ \ \upd(\texttt{Filter', fid'} \in \Pi_{\tt fid\_fk'}(\sigma_{\tt sid'=id}(\tt Subscriber')),$
\begin{verbatim}
      params', fltr)
  upd(Subscriber', sid'=id, sname', name)

List<Tuple> getSubName(int id)
\end{verbatim}
$\ \ \ \ \Pi_{\tt sname'}(\sigma_{\tt sid'=id}(\tt Subscriber'))$

\begin{verbatim}

List<Tuple> getSubFilter(int id)
\end{verbatim}
$\ \ \ \ \Pi_{\tt params'}(\sigma_{\tt sid'=id}(\tt Filter' \Join Subscriber'))$

\vspace{0.05in}
\caption{After Refactoring}
\end{subfigure}

\caption{Database Application in Intermediate Language.}
\label{fig:ir-ex}
\end{figure}

%% file: fig-semantics.tex
\begin{figure}
\centering
\[
\begin{array}{lcl}
    & & \\
    \multicolumn{3}{l}{\boxed{\denot{P} :: \emph{Invocation Sequence}~\omega \rightarrow ~\emph{Instance}~\instance \rightarrow \emph{List}}} \\
    & & \\
    \denot{(\S, \M, \Q)}_{(n, \sigma), \instance} &=& \emph{map}(\denot{Q_n}_{\sigma, \instance}, \lambda x. vals(x)) \\
    \denot{(\S, \M, \Q)}_{(n, \sigma); \omega, \instance} &=& \denot{(\S, \M, \Q)}_{\omega, \instance'} \ {\rm where} \  \instance' = \denot{U_n}_{ \sigma, \instance}  \\

    & & \\
    \multicolumn{3}{l}{\boxed{\denot{U} :: \emph{Valuation}~\sigma \rightarrow \emph{Instance}~\instance \rightarrow \emph{Instance}}} \\
    & & \\
    \denot{U_1; U_2}_{ \sigma, \instance} &=& \denot{U_2}_{ \sigma, \instance' } \ {\rm where} \  \instance' = \denot{U_1}_{ \sigma, \instance} \\
    \denot{\ins(R,t)}_{ \sigma, \instance} &=& \instance[R \leftarrow \emph{append}(\instance(R), t[\sigma])] 
    \\
    \denot{\del(R,\phi)}_{ \sigma, \instance} &=& \instance[R \leftarrow \emph{filter}(\instance(R), \lambda x. \neg {\denot{\phi}}_{ \sigma, \instance, x})]
    \\
    \denot{\upd(R,\phi,a,v)}_{\sigma, \instance} &=& \instance 
    \left [
      R \leftarrow \emph{append} \left (
    \begin{array}{l}
        \emph{filter}(\instance(R), \lambda x. {\neg \denot{\phi}}_{\sigma, \instance, x}), \\
        \emph{map}(\emph{filter}(\instance(R), \lambda x.  {\denot{\phi}}_{\sigma, \instance, x}), \lambda x. x[a \leftarrow v[\sigma]])
    \end{array}
    \right )
    \right ] \\

    & & \\
    \multicolumn{3}{l}{\boxed{\denot{Q} :: \emph{Valuation}~\sigma \rightarrow \emph{Instance}~\instance \rightarrow \emph{Relation}}} \\
    & & \\
    \denot{R}_{ \sigma, \instance} &=& \instance(R) \\
    \denot{\Pi_\psi(Q)}_{\sigma, \instance} &=& \emph{map}(\denot{Q}_{\sigma, \instance}, \lambda x. \emph{filter}(x, \lambda y. \emph{contains}(\emph{first}(y), \psi))) \\

    \denot{\sigma_\phi(Q)}_{ \sigma, \instance} &=& \emph{filter}( \denot{Q}_{ \sigma, \instance}, \lambda x. \denot{\phi}_{\sigma, \instance, x})\\
    \denot{Q_1 \times Q_2}_{ \sigma, \instance} &=& \emph{foldl} (\lambda ys. \lambda y. \emph{append}(ys, \emph{map}(\denot{Q_2}_{\sigma, \instance}, \lambda z. \emph{merge}(y,z))), \ \nil,  \ \denot{Q_1}_{ \sigma, \instance}) \\
    \denot{Q_1 \Join_\phi Q_2}_{ \sigma, \instance} &=& \denot{\sigma_\phi(Q_1 \times Q_2)}_{\sigma, \instance} \\
    \denot{Q_1 \cup Q_2}_{ \sigma, \instance} &=& \emph{append}(\denot{Q_1}_{\sigma, \instance}, \denot{Q_2}_{\sigma, \instance}) \\
    \denot{Q_1 - Q_2}_{\sigma, \instance} &=&  \emph{foldl} \ (\lambda ys. \lambda y. \emph{delete}(y, ys), \denot{Q_1}_{ \sigma, \instance},  \denot{Q_2}_{\sigma, \instance})

\\

    & & \\
\multicolumn{3}{l}{\boxed{\denot{\phi} :: \emph{Valuation}~\sigma \rightarrow \emph{Instance}~\instance \rightarrow \emph{Tuple}~x \rightarrow \emph{Predicate}}} \\
    & & \\
    \denot{a_1 \odot a_2}_{\sigma, \instance, x} &=& \emph{lookup}(x,a_1) \odot \emph{lookup}(x, a_2) \\
    \denot{a \odot v}_{\sigma, \instance, x} &=& \emph{lookup}(x, a) \odot v[\sigma] \\
    \denot{a \in Q}_{ \sigma, \instance, x} &=& \emph{contains}(\emph{lookup}(x,a), ~ \emph{map}(\denot{Q}_{ \sigma, \instance}, \lambda y. \emph{head}(\emph{vals}(y))) \\
    \denot{\phi_1 \land \phi_2}_{ \sigma, \instance, x} &=& \denot{\phi_1}_{\sigma, \instance, x} \land \denot{\phi_2}_{\sigma, \instance, x} \\
    \denot{\phi_1 \lor \phi_2}_{\sigma, \instance, x} &=& \denot{\phi_1}_{ \sigma, \instance, x} \lor \denot{\phi_2}_{ \sigma, \instance, x} \\
    \denot{\neg \phi}_{\sigma, \instance, x} &=& \neg \denot{\phi}_{ \sigma, \instance, x} \\

\end{array}
\]
\caption{Denotational semantics of database-driven applications. Our semantics are defined in terms of the standard list combinators $\emph{map}, \emph{append}, \emph{filter}$, $\emph{foldl}$, and $\emph{contains}$. Given a map  $x$, we write $\emph{vals}(x)$ to denote the list of values stored in the map, and we think of a map as a list of (key, value) pairs.    The function  $\emph{delete}(y, ys)$ removes the first occurence of $y$ in list $ys$.  Given two maps $y,z$ with disjoint keys, $\emph{merge}(y,z)$ generates a new map that contains all key-value pairs in $y$ and $z$. }
\label{fig:semantics}
\end{figure}

%% file: methodology.tex
\section{Proof Methodology}\label{sec:methodology}

Having defined the semantic equivalence and refinement properties for database-driven applications, let us now turn our attention to the proof methodology for showing these properties.

\subsection{Proving Equivalence}

A standard methodology for proving equivalence between any two systems $A, B$ is to find a \emph{bisimulation relation} that relates states in $A$ with those in $B$~{\cite{bisimulation}}. In our case, these systems are database-driven applications, and the states that we need to relate are database instances. Our approach does not directly infer an explicit mapping between database instances, but instead finds a \emph{bisimulation invariant} that (a) is satisfied by pairs of database instances from the two systems, and (b) is strong enough to prove equivalence.

In this paper, we prove that a bisimulation invariant $\Phi$ is valid by showing that it is inductive. That is, $\Phi$ must hold initially, and  assuming that it holds for a pair of databases $\instance, \instance'$, it must continue to hold after executing any pair of corresponding update operations $\lamvec{x}{U}$ and $\lamvec{y}{U'}$.

\vspace{0.05in}
\begin{definition} \textbf{(Inductive bisimulation invariant)}\label{def:inductive} Consider programs $P = (\S, \M, \Q)$ and $P' = (\S', \M', \Q')$ and suppose that $P, P'$ contain a disjoint set of variables (which can be enforced using $\alpha$-renaming if necessary).
    A bisimulation invariant $\Phi$ is said to be \emph{inductive} with respect to programs $P$ and $P'$ if (a) $\Phi$ is satisfied by a pair of empty databases, and (b) the following Hoare triple is valid for all $\lamvec{x}{U_i} \in \M$ and $\lamvec{y}{U_i'} \in \M'$:
\[ 
\{\Phi \land \vec{x} = \vec{y}\} \ \ U_i \ || \  U'_i \ \ \{\Phi\}
\]
\end{definition}

In the above definition, the notation $U || U'$ denotes the parallel execution of updates $\lamvec{x}{U}$ and $\lamvec{y}{U'}$. However, since  programs $P, P'$ contain a disjoint set of variables, $U || U'$ is semantically equivalent to the sequential composition $U ; U'$. Thus, to prove inductiveness, we need to show the validity of the Hoare triple
\[
    \{\Phi \land \vec{x} = \vec{y} \} \ U_i \ ; \ U'_i \ \{\Phi\}
\]
for every pair of updates $\lamvec{x}{U_i}$ and $\lamvec{y}{U_i'}$ in $P$ and $P'$. Also, observe that $\Phi$ must hold for a pair of empty databases in the base case because we assume that the databases are initially empty (recall Section~\ref{sec:formalism}).~\footnote{This assumption is realistic in situations where database migration is performed by calling the new update methods in the application. Otherwise, the base case needs to establish that $\Phi$ holds for the initial databases.}

While there are many possible inductive bisimulation invariants (including \emph{true}, for example), we need a bisimulation invariant that is strong enough to prove equivalence. According to Definition~\ref{def:equiv}, two programs are equivalent if they yield the same result for every pair of corresponding queries $\lamvec{x}{Q}$ and $\lamvec{y}{Q'}$ given the same input. Thus, we can define what it means for a bisimulation invariant to be sufficient in the following way:

\vspace{0.05in}
\begin{definition} \textbf{(Sufficiency)}\label{def:sufficiency}
    A formula $\Phi$ is said to be \emph{sufficient} with respect to programs $P = (\S, \M, \Q)$ and $P' = (\S', \M', \Q')$ if,  for all $\lamvec{x}{Q_i} \in \Q$ and $\lamvec{y}{Q'_i} \in \Q'$, we have:
    \[ (\Phi \land \vec{x} = \vec{y}) \models   Q_i = Q_i' \]
\end{definition}

Our general proof methodology for proving equivalence is to find a \emph{sufficient, inductive bisimulation invariant} between the given pair of programs. If we can find such an invariant $\Phi$, we know that $\Phi$ holds after executing any invocation sequence $\omega$ on $P, P'$, so $\Phi$ must also hold before issuing any database query. Furthermore, since $\Phi$ is a sufficient bisimulation invariant, it implies that any pair of queries yield the same result. Thus, the existence of such a bisimulation invariant $\Phi$ implies that $P,P'$ are semantically equivalent.

\vspace{0.05in}
\begin{theorem} \textbf{(Soundness)}\label{thm:soundness}
Given database applications $P, P'$, the existence of a sufficient, inductive bisimulation invariant $\Phi$ implies $P \simeq P'$.
\end{theorem}

\begin{proof}
We show that  $\denot{P}_{\omega} = \denot{P'}_{\omega}$ by performing induction on the length of the invocation sequence $\omega$. However, we consider the following strengthened inductive hypothesis: ``If $\Phi$ is satisfied by database instances $\Delta, \Delta'$, then we have $\denot{P}_{\omega, \instance} = \denot{P'}_{\omega, \instance'}$''. Since $\denot{P}_{\omega}$ is the same as $\denot{P}_{\omega, \emptyset}$, and $\Phi$ is satisfied by a pair of empty databases according to Definition~\ref{def:inductive}, this implies $\denot{P}_{\omega} = \denot{P'}_{\omega}$.

For the base case, we have $\omega = (n, \sigma)$. By assumption, $\Delta \cup \Delta'$ is a model of $\Phi$; thus, $\Phi$ holds initially.  Since we evaluate queries $Q_n$ and $Q_n'$ on the same valuation $\sigma$, we also have $\vec{x} = \vec{y}$. Because $\Phi$ is a \emph{sufficient} bisimulation invariant, this implies $Q_i = Q_i'$; hence we have $\denot{P}_{\omega, \instance} = \denot{P'}_{\omega, \instance'}$ for the base case. For the inductive step, suppose the invocation sequence is of the form $(n, \sigma); \omega'$, where $\omega'$ is non-empty. In this case, we have $\denot{P}_{\omega, \instance} = \denot{P}_{\omega', \instance_1}$ and $\denot{P'}_{\omega, \instance'} = \denot{P'}_{\omega', \instance_2}$ where $\instance_1 = \denot{U_n}_{ \sigma, \instance}$ and $\instance_2 = \denot{U_n'}_{ \sigma, \instance'}$. Since we evaluate $U_n$ and $U_n'$ on the same valuation $\sigma$, and since $\Delta \cup \Delta'$ is a model of $\Phi$, we have $\Phi \land \vec{x} = \vec{y}$. Using Definition~\ref{def:inductive}, we therefore know that $\Delta_1 \cup \Delta_2$ is also a model of $\Phi$. Thus, the theorem follows immediately from the inductive hypothesis.
\end{proof}

\begin{theorem} \textbf{(Relative Completeness)}\label{thm:completeness}
 Suppose we have an oracle for  proving any valid Hoare triple and  logical entailment.
If $P \simeq P'$, then there always exists a sufficient, inductive bisimulation invariant $\Phi$ for programs $P, P'$. 
\end{theorem}
\vspace{0.05in}

%\begin{proof}[Proof Sketch]
\textsc{Proof Sketch.}
Recall that we have $P \simeq P'$ iff $\denot{P}_{\omega} = \denot{P'}_{\omega}$ for an \emph{arbitrary} invocation sequence $\omega$. Because any valid invocation sequence starts with arbitrarily many update transactions, followed by a single query, we have $\denot{P}_{\omega} = \denot{P'}_{\omega}$ iff the imperative program shown in Figure~\ref{fig:complete} is safe. Essentially, the  program in Figure~\ref{fig:complete} picks a random number of update transactions (together with a randomly chosen valuation), followed by a single, but arbitrary query transaction (and its corresponding randomly chosen valuation).

\input{fig-complete}

If the assertion of program in Figure~\ref{fig:complete} is valid, the Hoare triple
\[
    \hspace{-2.7in}
    \hoare{true}{S_0; ~ while(*) S_1; ~ S_2}{R = R'}
\]
must be provable (by relative completeness of Hoare logic, and assuming an oracle for proving Hoare triples that involve database transactions). Again, using relative completeness of Hoare logic, this means the following three Hoare triples must be valid, where $I$ is an inductive invariant of the while loop:

\begin{enumerate}
    \item $\hoare{true}{S_0}{I}$
    \item $\hoare{I}{S_1}{I}$
    \item $\hoare{I}{S_2}{R = R'}$
\end{enumerate}

(1) implies that $I$ is satisfied by a pair of empty databases, since $S_0$ just initializes $\instance$ and $\instance'$ to be empty.
(2) implies Hoare triple $\hoare{I \land \vec{x}=\vec{y}}{U_i;U'_i}{I}$ is valid since the loop body $S_1$ executes $U_i$ and $U_i'$ on the same valuation and $\vec{x}, \vec{y}$ refer to the parameters of $U_i, U_i'$, respectively.
(3) implies Hoare triple $\hoare{I \land \vec{x}=\vec{y}}{Q_i; Q_i'}{R = R'}$ is valid since we again execute $Q_i, Q_i'$ on the same valuation. Thus, we have $I \land \vec{x}=\vec{y} \models Q_i = Q_i'$ because $R, R'$ are the return values of $Q_i, Q_i'$. Therefore, we have shown that $I$ is a sufficient, inductive bisimulation invariant of $P$ and $P'$.
\hfill $\qed$
%\end{proof}

\subsection{Proving Refinement}

Since our notion of refinement is a generalization of equivalence (recall Definition~\ref{def:ref}), our proof methodology for showing program refinement closely follows that for verifying equivalence. In particular, rather than finding a one-to-one mapping between database states as in the case of equivalence, it suffices to find a one-to-many mapping for showing refinement. Hence, our proof methodology relies on finding a \emph{simulation invariant} rather than a stronger bisimulation invariant:

\begin{definition} \textbf{(Inductive simulation invariant)}\label{def:ind-sim}
Consider programs $P = (\S, \M, \Q)$ and $P' = (\S', \M', \Q')$ and suppose that $P, P'$ contain a disjoint set of variables. A \emph{simulation invariant} $\Phi$ is said to be \emph{inductive} with respect to programs $P$ and $P'$ if (a) $\Phi$ is satisfied by a pair of empty databases, and (b) the following Hoare triple is valid for  all $\lamvec{x}{U_i} \in \M$ and $\lamvec{y}{U_i'} \in \M'$ where $i \in [1, |\M|]$:
\[
\Big\{ \Phi \land \bigwedge_{x_j \in \vec{x}} x_j= y_j\Big\} \ \ U_i \ ; \ U'_i \ \ \Big \{\Phi\Big\}
\]

\end{definition}

Recall from Definition~\ref{def:ref} that we allow the second program $P'$ to contain more transactions than $P$, but the notion of refinement only talks about invocation sequences that use shared transactions from $P$ and $P'$. Therefore, in Definition~\ref{def:ind-sim}, we only require $\Phi$ to be preserved by pairs of update transactions that are both present in $P$ and $P'$. Furthermore, since transactions in $P'$ can take additional arguments not present in their counterparts in $P$, our precondition states that the arguments are pairwise equal for only the ``shared'' variables.

As in the equivalence scenario, finding an inductive simulation invariant $\Phi$ between $P$ and $P'$ is \emph{not} sufficient for proving that $P'$ refines $P$, as $\Phi$ may not be strong enough to show refinement. Hence, we also need to define what it means for an inductive simulation invariant to be \emph{sufficient} for showing refinement. However, since the notion of refinement is weaker than equivalence, we also weaken our corresponding notion of sufficiency as follows:

\begin{definition} \textbf{(Projective sufficiency)}
A formula $\Phi$ is said to be \emph{projectively sufficient} with respect to programs $P = (\S, \M, \Q)$ and $P' = (\S', \M', \Q')$ if  for all $\lamvec{x}{Q_i} \in \Q$ and $\lamvec{y}{Q'_i} \in \Q'$ where $i \in [1, |\Q|]$, we have:
\[ \Big ( \Phi \land \bigwedge_{x_j \in \vec{x}} x_j = y_j \Big ) \models   \exists L. ~ Q_i = \Pi_L(Q_i') \]
\end{definition}

Observe that the notion of projective sufficiency is weaker than Definition~\ref{def:sufficiency}, as we do not require $Q_i$ and $Q_i'$ to yield exactly the same relation and allow the result of $Q_i'$ to contain attributes not present in the result of $Q_i$. Our general proof methodology for proving refinement then relies on finding a simulation invariant that is both inductive and projectively sufficient.

\begin{theorem} \textbf{(Soundness)}
Given database applications $P, P'$, the existence of a projectively sufficient and inductive simulation invariant $\Phi$ implies $P' \preceq P$.
\end{theorem}

\begin{theorem} \textbf{(Relative Completeness)}
Suppose we have an oracle for proving any valid Hoare triple and logical entailment.
If $P' \preceq P$, then there always exists a projectively sufficient and inductive simulation invariant $\Phi$ for programs $P, P'$.
\end{theorem}

The proofs of these theorems are very similar to those of Theorems~\ref{thm:soundness} and~\ref{thm:completeness} for  equivalence.

%% file: fig-complete.tex
\begin{wrapfigure}{r}{0.49\textwidth}
\begin{minipage}{0.49\textwidth}
\[
\boxed{
\begin{array}{l}
\textsc{IsEquivalent}(P, P') \ \{
\\%
\left .
\begin{array}{l}
\quad {\rm assume}(P = (S, \vec{T}_U, \vec{T}_Q));  \\
\quad {\rm assume}(P' = (S', \vec{T}'_U, \vec{T}'_Q)); \\
\quad \instance = \emptydb; \ \ \instance' = \emptydb;
\end{array}
\quad \right ]  S_0
\\%
\ \ \ \quad {\rm while}(*) \{ \\
\left .
\begin{array}{l}
\quad \quad \quad i := {\rm randInt}(1, |\vec{T_U}|); \\
\quad \quad \quad \sigma := {\rm randValuation}(U_i); \\
\quad \quad \quad \instance := U_i(\sigma, \instance); \\
\quad \quad \quad \instance' := U_i'(\sigma, \instance'); \\
\end{array}
\right ] S_1 \\
\ \ \  \quad \}
\\%
\left .
\begin{array}{l}
\  \quad n := {\rm randInt}(1, |\vec{T_Q}|); \\
\  \quad \sigma := {\rm randValuation}(Q_n); \\
\  \quad R := Q_n(\sigma, \instance); \\
\  \quad R' := Q'_n(\sigma, \instance'); \\
\end{array}
\ \ \quad \right ] S_2 \\
\ \ \ \quad {\rm assert}(R = R'); \\
\}
\end{array}
}
\]
\end{minipage}
\caption{Program that is safe iff $P,P'$ are equivalent.}
\vspace{0.05in}
\label{fig:complete}
\end{wrapfigure}

%% file: theory.tex
\section{SMT Encoding of Relational Algebra with Updates}\label{sec:logic}
In the previous section, we defined what it means for simulation and bisimulation invariants to be inductive, but we have not fixed a logical theory over which we express these invariants. In this section, we discuss the theory of relational algebra with updates, $\tra$, and show how to enable reasoning in $\tra$ using existing SMT solvers.

\input{fig-theory}

Figure~\ref{fig:theory} gives the syntax of the theory of relational algebra with updates, $\tra$, which we use to express simulation and bisimulation invariants . Atomic formulas  in $\tra$ are of the form $t = t$ where $t$ is a term representing a relation. Basic terms include variables $x$ and concrete tables $T$, and more complex terms can be formed using the relational algebra operators $\Pi$ (projection), $\sigma$ (selection), $\times$ (Cartesian product),   $\cup$ (union), and $-$ (difference). In addition to these standard relational algebra operators, $\tra$ also includes an \emph{update} operator, denoted as $t\awr{a_i}{v}$, which represents the new relation after changing the  $i$'th attribute of all tuples in $t$ to $v$. Observe that the  theta join operator $\Join_\phi$ is expressible in this logic as $\sigma_\phi(t_1 \times t_2)$. We also write $t_1 \Join t_2$ as syntactic sugar for $\sigma_{\phi}(t_1 \times t_2)$   where $\phi$ is a predicate stating that {the shared attributes of $t_1$ and $t_2$ are equal}.

\input{fig-axiom}

\input{fig-axiom-aux}

Since we view tables as lists of tuples, we axiomatize  $\tra$ using the theory of lists~\cite{smtlib}. Our axiomatization is presented in  Figure~\ref{fig:axiom} in the form of inference rules, where we view tuples as \emph{lists} of values and  relations as \emph{lists} of tuples. An attribute $a_i$ of a tuple is simply an index into the list representing that tuple.
For example, consider the axioms for \emph{projection} $\Pi_l(t)$, which projects term $t$ given attribute list $l$. We first define an auxiliary function $\Pi'_l(h)$ that projects a single tuple $h$ given $l$. In particular, if the attribute list $l$ is empty, then $\Pi'_l(h)$ yields $\nil$. Otherwise, if $l$ consists of head $a_i$ and tail $l_1$, $\Pi'_l(h)$ composes the $i$'th value of $h$ (i.e., $\texttt{get}(i, h)$) and the projection of tail $\Pi'_{l_1}(h)$. Similarly, $\Pi_l(t)$ is also recursively defined. If $t = \nil$, then $\Pi_l(t) = \nil$. Otherwise if $t = h :: t_1$, then $\Pi_l(t)$ composes the projection $\Pi'_l(h)$ of head $h$ and the projection $\Pi_l(t_1)$ of tail $t_1$. Also,
please observe that the inference rules for selection in Figure~\ref{fig:axiom} actually correspond to \emph{axiom schemata} rather than axioms: Because the selection operator is parameterized over a predicate $\phi$, this schema needs to be instantiated for each predicate that occurs in the formula.

\begin{example}
Consider the formula $\sigma_{a_1 \geq 2}(x) = \sigma_{a_2 > 1}(y)$ in the theory of relational algebra with updates. We generate the following axioms using the axiom schemata for selection:
\[
\begin{array}{ll}
(1a) & \forall x.  \ (x = [ \ ]) \rightarrow \sigma_{a_1 \geq 2}(x) = [ \ ] \\
(1b) & \forall x, h, t. \ {\Big (} x=h::t \rightarrow  \Big (
\begin{array}{l}  (\emph{get}(1, h) \geq 2 \rightarrow \sigma_{a_1 \geq 2}(x)= h::\sigma_{a_1 \geq 2}(t)) \land \\ (\neg (\emph{get}(1, h) \geq 2) \rightarrow \sigma_{a_1 \geq 2}(x)= \sigma_{a_1 \geq 2}(t))
\end{array} \Big ) {\Big )}
\\
(2a) & \forall y.  \ (y = [ \ ]) \rightarrow \sigma_{a_2 > 1}(y) = [ \ ] \\
(2b) & \forall y, h, t. \ {\Big (} y=h::t \rightarrow  \Big (
\begin{array}{l}  (\emph{get}(2, h) > 1 \rightarrow \sigma_{a_2 > 1}(y)= h::\sigma_{a_2 > 1}(t)) \land \\ (\neg (\emph{get}(2, h) > 1) \rightarrow \sigma_{a_2 > 1}(y)= \sigma_{a_2 > 1}(t))
\end{array} \Big ) {\Big )}

\end{array}
\]
\end{example}

\vspace{0.1in} \noindent
{\bf \emph{Remark.}} Since the problem of checking equivalence between a pair of relational algebra expressions is known to be undecidable~\cite{undecidable}, our theory of relational algebra with updates is also undecidable. However, with the aid of some optimizations that we discuss in Section~\ref{sec:impl}, we are able to determine the validity of most $\tra$ formulas that we encounter in practice.

%% file: fig-theory.tex
\begin{figure}
\centering
\[
\begin{array}{r c l}
    \emph{Formula} \ \ \ F &:=& \emph{true} \termdelim \emph{false} \termdelim  t = t \termdelim  F \land F \termdelim F \lor F \termdelim \neg F \termdelim F \rightarrow F \termdelim \exists x. F \termdelim \forall x. F \\
    \emph{Term} \ \ \ t &:=& x \termdelim T \termdelim \Pi_{a_1,\ldots,a_n}(t) \termdelim \sigma_{\phi}(t) \termdelim t \times t \termdelim t \cup t \termdelim t - t \termdelim t \awr{a_i}{v} \\
    \emph{Predicate} \ \ \ \phi &:=& a_i \odot a_j \termdelim a_i \odot v \termdelim a_i \in t \termdelim \phi \land \phi \termdelim \phi \lor \phi \termdelim \neg \phi \\
    \emph{BinOp} \ \ \ \odot &:=& \leq \termdelim < \termdelim = \termdelim \neq \termdelim > \termdelim \geq \\
\end{array}
\]
\vspace{-0.05in}
\[
\begin{array}{r c l c r c l}
    T & \in & \emph{Table} & \ \ \ \ \ \ \ \ \ & a_i & \in & \emph{Attribute} \\
    v & \in & \emph{Constant} \cup \emph{Variable} & & x & \in & \emph{Variable} \\
\end{array}
\]
\caption{Formula in Theory of Relational Algebra with Updates. The notation $a_i$ represents the $i$'th attribute in a relation.}
\label{fig:theory}
\end{figure}

%% file: fig-axiom.tex
\begin{figure}
\centering

\begin{subfigure}[b]{.49\linewidth}
\centering
\boxed{\emph{get \ $\texttt{get}(i,l)$}}
\[
    \irule{l = h :: t \quad i = 0}{\texttt{get}(i,l) = h} \ruledelim
    \irule{l = h :: t \quad i \neq 0}{\texttt{get}(i,l) = \texttt{get}(i-1,t)}
\]

\vspace{0.10in}

\boxed{\emph{projection $\Pi_l(t)$}}
\[
    \irule{t = [~]}{\Pi_l(t) = [~]} \ruledelim
    \irule{t = h :: t_1}{\Pi_l(t) = \Pi'_l(h) :: \Pi_l(t_1)}
\]
\[
    \irule{l = [~]}{\Pi'_l(h) = [~]} \ruledelim
    \irule{l = a_i :: l_1}{\Pi'_l(h) = \texttt{get}(i,h) :: \Pi'_{l_1}(h)}
\]

\vspace{0.10in}

\boxed{\emph{union $t_1 \cup t_2$}}
\[
    \irule{t_1 = [~]}{t_1 \cup t_2 = t_2} \ruledelim
    \irule{t_1 = h :: t}{t_1 \cup t_2 = h :: (t \cup t_2)}
\]

\vspace{0.10in}

\boxed{\emph{minus $t_1 - t_2$}}
\[
    \irule{t_2 = [~]}{t_1 - t_2 = t_1} \ruledelim
    \irule{t_2 = h_2 :: t}{t_1 - t_2 = (t_1 -' h_2) - t}
\]
\[
    \irule{t_1 = [~]}{t_1 -' h_2 = [~]} \ruledelim
    \irule{t_1 = h_1 :: t_3 \quad h_1 = h_2}{t_1 -' h_2 = t_3}
\]
\[
    \irule{t_1 = h_1 :: t_3 \quad h_1 \neq h_2}{t_1 -' h_2 = h_1 :: (t_3 -' h_2)}
\]
\end{subfigure}
~
\begin{subfigure}[b]{.49\linewidth}
\centering
\boxed{\emph{selection $\sigma_\phi(t)$}}
\[
    \irule{t = [~]}{\sigma_\phi(t) = [~]} \ruledelim
    \irule{\phi(h) \quad t = h :: t_1}{\sigma_\phi(t) = h :: \sigma_\phi(t_1)}
\]
\[
    \irule{\neg \phi(h) \quad t = h :: t_1}{\sigma_\phi(t) = \sigma_\phi(t_1)}
\]

\vspace{0.10in}

\boxed{\emph{Cartesian product $t_1 \times t_2$}}
\[
    \irule{t_1 = [~]}{t_1 \times t_2 = [~]} \ruledelim
    \irule{t_1 = h_1 :: t}{t_1 \times t_2 = (h_1 \times' t_2) \cup (t \times t_2)}
\]
\[
    \irule{t_2 = [~]}{h_1 \times' t_2 = [~]} \ruledelim
    \irule{t_2 = h_2 :: t_3}{h_1 \times' t_2 = \texttt{cat}(h_1, h_2) :: h_1 \times' t_3}
\]
\[
    \irule{h_1 = [~]}{\texttt{cat}(h_1, h_2) = h_2} \ruledelim
    \irule{h_1 = c :: h}{\texttt{cat}(h_1, h_2) = c :: \texttt{cat}(h,h_2)}
\]

\vspace{0.10in}

\boxed{\emph{update $t \awr{a_i}{v}$}}
\[
    \irule{t = [~]}{t \awr{a_i}{v} = [~]} \ \ \
    \irule{t = h :: t_1}{t \awr{a_i}{v} = \texttt{upd}(h,i,v) :: t_1 \awr{a_i}{v}}
\]
\[
    \irule{h = [~]}{\texttt{upd}(h,i,v) = [~]} \ruledelim
    \irule{h = c :: h_1 \quad i = 0}{\texttt{upd}(h,i,v) = v :: h_1}
\]
\[
    \irule{h = c :: h_1 \quad i \neq 0}{\texttt{upd}(h,i,v) = c :: \texttt{upd}(h_1, i-1, v)} \ruledelim
\]
\end{subfigure}

\caption{Axioms in Theory of Relation Algebra with Updates. $[~]$ represents an empty list \texttt{nil}. The binary operator $::$ denotes the list constructor \texttt{cons}. $\times'$ and $\texttt{cat}$ are auxiliary functions for axiomatizing Cartesian product. $\Pi'$, $-'$, and $\texttt{upd}$ are auxiliary functions for axiomatizing projection, minus, and update, respectively.}
\label{fig:axiom}
\end{figure}

%% file: fig-axiom-aux.tex
\begin{figure}
\centering
\[
\begin{array}{rcl}
(a_i \odot a_j)(h) &=& \texttt{get}(i, h) \odot \texttt{get}(j, h) \\
(a_i \odot v)(h) &=& \texttt{get}(i, h) \odot v \\
(a_i \in t)(h) &=& \exists j. ~ \texttt{get}(0, \texttt{get}(j, t)) = \texttt{get}(i, h) \\
(\phi_1 \land \phi_2)(h) &=& \phi_1(h) \land \phi_2(h) \\
(\phi_1 \lor \phi_2)(h) &=& \phi_1(h) \lor \phi_2(h) \\
(\neg \phi)(h) &=& \neg \phi(h) \\
\end{array}
\]
\caption{Auxilary functions for selection axiom schema $\phi(h)$}
\label{fig:axiom-aux}
\end{figure}

%% file: automation.tex
\section{Automated Verification}

So far, we have explained our general proof methodology and introduced a first-order theory in which we will express our  bisimulation invariants. However, we have not yet explained how to \emph{automatically} prove equivalence between programs. In this section, we discuss our strategy for proof automation. Specifically, we first discuss how to automatically prove equivalence assuming that an oracle provides bisimulation invariants (Section~\ref{sec:sp}), and then explain how we infer them automatically (Section~\ref{sec:inv}). Because the automation of refinement checking is very similar, this section only addresses equivalence.

\subsection{Automation for Bisimulation Invariant Inductiveness}\label{sec:sp}

Consider two database-driven programs $P = (\S, \M, \Q)$ and $P' = (\S', \M', \Q')$, and suppose that an oracle provides a bisimulation invariant $\Phi$ between $P$ and $P'$. Based on the proof methodology we outlined in Section~\ref{sec:methodology}, we can prove that $P, P'$ are equivalent by showing that $\Phi$ satisfies the following conditions for any pair of updates $\lambda \vec{x}. U_i$,  $\lambda \vec{y}. U_i'$ and any pair of queries $\lambda \vec{x}. Q_i$,  $\lambda \vec{y}. Q_i'$:

\[
\begin{array}{lll}
(1) &  \ \Phi \land \vec{x} = \vec{y} \models Q_i = Q_i'&  \emph{(Sufficiency)}  \\
(2) & \{ \Phi \land \vec{x} = \vec{y} \} \ U_i; U_i' \ \{ \Phi \} & \emph{(Inductiveness)}
\end{array}
\]

The first condition (sufficiency) is easy to prove since we have already defined a logical theory that allows us to write terms of the form $Q_i = Q_i'$. The only small technical hiccup is that $\tra$ uses the syntax $a_i$ to denote the $i$'th attribute in a relation, whereas attributes in the queries $Q_i, Q_i'$ are \emph{names} of attributes. To solve this difficulty, we assume a function $\attnum$ which replaces attribute names $s$ in constructs from Figure~\ref{fig:ir} with  $a_i$, where $i$ is the index of $s$. Thus, we can check whether formula $\Phi$ satisfies condition (1) by querying whether the following formula is valid modulo $\tra$:
\[
(\Phi \land \vec{x} = \vec{y}) \rightarrow \attnum(Q_i) = \attnum(Q_i')
\]

However, to prove the second condition (i.e., inductiveness),  we need a way to prove Hoare triples for update statements $U$ from Figure~\ref{fig:ir}. Towards this goal, we define a \emph{strongest post-condition} semantics for update statements. Given a formula $\Phi$ over $\tra$ and an update statement $U$, Figure~\ref{fig:sp} describes the computation of $\emph{sp}(\Phi, U)$, which represents the strongest post-condition of $\Phi$ with respect to statement $U$.

\input{fig-sp}

To compute the strongest postcondition of $\Phi$ with respect to $\ins(R, \{ a_1: v_1, \ldots, a_n: v_n\})$, we think of the insertion as the assignment $R := \emph{append}(R, [r ])$ where $r$ is the tuple (list) $[v_1, \ldots, v_n]$.  Since the union operator $\cup$ in $\tra$ corresponds to list concatenation, the new value of $R$ after the assignment is given by $x \cup [r]$, where the existentially quantified variable $x$ represents the old value of $R$.

To understand the strongest postcondition semantics of deletion, recall that $\del(R, \phi)$ removes all rows in $R$ that satisfy $\phi$. Hence, we can model this statement using the assignment $R:= \sigma_{\neg \phi}(R)$. Thus, when we compute the strongest postcondition of $\phi$ with respect to $\del(R, \phi)$, the new value of $R$ is given by $\sigma_{\attnum(\neg \phi)}(x)$ where the existentially quantified variable $x$ again represents the old value of $R$ and $\attnum(\phi)$ replaces attribute names in $\phi$ with their corresponding indices.

Finally, let us consider the strongest postcondition for update statements of the form $\upd(R, \phi, a, v)$. Recall that this statement assigns value $v$ to the $a$ attribute of all tuples in $R$ that satisfy $\phi$. Specifically, according to the denotational semantics from Figure~\ref{fig:semantics}, we can model $\upd(R, \phi, a, v)$ using the assignment statement:
\[
R \ := \ (\sigma_{\neg \phi}(R)) \cup (\sigma_\phi(R))\awr{a}{v}
\]
Hence, we obtain the strongest postcondition of $\Phi$ with respect to $\upd(R, \phi, a, v)$ by computing the strongest postcondition of the above assignment, where the right-hand side is a term in $\tra$ (modulo changing attribute names to indices).

\begin{definition}{\bf (Agreement)}
Consider database instance $\instance$ and valuation $\sigma$, and let $\attnum(\instance)$ denote the representation of $\instance$ where each tuple $\{ a_1: v_1, \ldots, a_n: v_n \}$ is represented as the list $[v_1, \ldots, v_n]$. We say that $(\instance, \sigma)$ agrees with $\tra$-formula $\Phi$, written $(\instance, \sigma) \sim \Phi $, iff $\attnum(\instance) \uplus \sigma \models \Phi$.
\end{definition}

\begin{theorem}{\bf (Soundness of sp)}\label{thm:sp}
Suppose that $\denot{U}_{\sigma, \instance} = \instance'$, and let $\Phi$ be a $\tra$ formula. If $(\instance, \sigma) \sim \Phi$, then we also have
$(\instance', \sigma) \sim \emph{sp}(\Phi, U)$.
\end{theorem}

\begin{proof}
See Appendix A.
\end{proof}

Now that we have defined a strongest post-condition semantics for update statements in our language, it is easy to check the correctness of the Hoare triple $\{ \Phi \land \vec{x} = \vec{y} \} \ U_i; U_i' \ \{ \Phi \}$ by simply querying the validity of the following formula modulo $\tra$:
\[
\emph{sp}(\Phi \land \vec{x} =\vec{y}, U_i; U_i') \rightarrow \Phi
\]

\subsection{Bisimulation Invariant Synthesis}\label{sec:inv}

So far, we have discussed how to automate the proof that $\Phi$ is an inductive and sufficient bisimulation invariant. However, since we do not want users to manually provide such bisimulation invariants, our verification algorithm automatically infers them using \emph{monomial predicate abstraction}~\cite{mpa1,mpa2,mpa3}.

\input{algo-verify}

Our technique for inferring suitable bisimulation invariants is shown in Algorithm~\ref{algo:verify}. Given two programs $P, P'$ with corresponding schemas $S, S'$, the {\sc Verify} procedure first generates the universe of all predicates that may be used in the bisimulation invariant (line 4). We generate all such predicates by instantiating the following database of predefined templates:
\begin{enumerate}
\item $\Pi_L ( ~ R ~ ) = \Pi_{L'} ( ~ R' ~)$
\item $\Pi_L ( ~ R_1 \Join R_2 ~ ) = \Pi_{L'} ( ~ R_1' ~)$
\item $\Pi_L ( ~ R ~ ) = \Pi_{L'} ( ~ R_1' \Join R_2' ~)$
\item $\Pi_L ( ~ R_1 \Join R_2 ~ ) = \Pi_{L'} ( ~ R_1' \Join R_2' ~)$
\end{enumerate}

In these templates, $L$ and $R$ represent an attribute list and a relation under schema $\S$, while $L'$ and $R'$ represent the attribute list and relation under schema $\S'$. Please note that we only consider  templates with at most one join operator on each side. Any predicate containing a longer join chain can be decomposed into several predicates of these forms.

Once we generate the universe $\universe$ of all predicates that may be used in the invariant, we perform a fixed point computation in which we iteratively \emph{weaken} the candidate bisimulation invariant. Initially, the candidate bisimulation invariant $\Phi$ starts out as the conjunction of all predicates in our universe. During the fixed point computation (lines 6--12 in Algorithm~\ref{algo:verify}), the candidate invariant $\Phi$ is always stronger than the actual bisimulation invariant. Hence, if we get to a point where $\Phi$ is not strong enough to show equivalence, we conclude that the program cannot be verified using conjunctive formulas over our templates (line 13). On the other hand, assuming that  $\Phi$ is strong enough to prove equivalence, we then proceed to check whether $\Phi$ is inductive (lines 7--12). If it is,   the strongest postcondition of $\Phi \land \vec{x} = \vec{y}$ must logically imply $\varphi$ for every predicate $\varphi$ used in $\Phi$. If some predicate $\varphi$ is not preserved by a pair of updates (i.e., call to \textsf{CheckInductiveness} returns false), we then remove $\varphi$ from both $\Phi$ and our universe of predicates $\universe$. We continue this process of weakening the invariant until it becomes an inductive bisimulation invariant, or we prove that no such invariant exists over our universe of predicates.

%% file: fig-sp.tex
\begin{figure}
\centering
\[
\begin{array}{rcl}
\emph{sp}(\Phi, \ins(R, \{ a_1: v_1, \ldots, a_n: v_n\})) &=& \exists x. ~ (R = x \cup [r]) \land \Phi[x/R] \ \emph{where} \ r = [v_1, \ldots, v_n]\\
\emph{sp}(\Phi, \del(R, \phi)) &=& \exists x. ~ (R = \sigma_{\attnum(\neg \phi)}(x)) \land \Phi[x/R] \\
\emph{sp}(\Phi, \upd(R, \phi, a, v)) &=& \exists x. ~ (R = \sigma_{\attnum(\neg \phi)}(x) \cup \sigma_{\attnum(\phi)}(x)\awr{\attnum(a)}{v}) \land \Phi[x/R] \\
\emph{sp}(\Phi, U_1; U_2) &=& \emph{sp}(\emph{sp}(\Phi, U_1), U_2) \\
\end{array}
\]
\caption{Strongest Postcondition for Update Transactions}
\label{fig:sp}
\end{figure}

%% file: algo-verify.tex
\algnewcommand\algorithmicforeach{\textbf{for each}}
\algdef{S}[FOR]{ForEach}[1]{\algorithmicforeach\ #1\ \algorithmicdo}

\begin{algorithm}[t]
\caption{Verification Algorithm for Database Application Equivalence}
\label{algo:verify}
\begin{algorithmic}[1]

\vspace{0.05in}
\Procedure{\sc Verify}{$P$, $P'$}
\vspace{0.05in}
\State {\rm \bf Input:} program $P = (\S, \M, \Q)$ and $P' = (\S', \M', \Q')$
\State {\rm \bf Output:} Bisimulation invariant $\Phi$ establishing equivalence or $\bot$ to indicate failure
\vspace{0.05in}

\State $\universe$ := \textsf{GetAllPredicates}($\S$, $\S'$)
\State $\Phi$ := $\bigwedge_{\varphi \in \universe} \varphi$ 
\While{\textsf{CheckSufficiency}($\Phi$, $\Q$, $\Q'$)}
    \State ind := $\emph{true}$
    \ForEach{$\varphi \in \universe$}
        \If{$\neg$\textsf{CheckInductiveness}($\Phi$, $\M$, $\M'$, $\varphi$)}
                    \State ind := $\emph{false}$;  \  $\universe$ := $\universe \setminus \{ \varphi \}$;  \ $\Phi$ := $\bigwedge_{\varphi \in \universe} \varphi$;
            \State \textbf{break}
        \EndIf
    \EndFor
    \If{ind}{~ \Return $\Phi$;} \EndIf
\EndWhile
\State \Return $\bot$;
\EndProcedure

\vspace{0.1in}
\Procedure{\sc CheckSufficiency}{$\Phi$, $\Q$, $\Q'$}
\ForEach{$\lamvec{x}{Q_i} \in \Q$ and $\lamvec{y}{Q_i'} \in \Q'$}
    \If{ $\tra \not \models (\Phi \land \vec{x} = \vec{y} \to \attnum(Q_i) = \attnum(Q_i'))$}{~ \Return $\emph{false}$;} \EndIf
\EndFor
\State \Return $\emph{true}$;
\EndProcedure

\vspace{0.1in}
\Procedure{\sc CheckInductiveness}{$\Phi$, $\M$, $\M'$, $\varphi$}
\ForEach{$\lamvec{x}{U_i} \in \M$ and $\lamvec{y}{U_i'} \in \M'$}
\If{$\tra \not \models (sp(\Phi \land \vec{x} = \vec{y}, ~ {U_i}; {U_i'}) \to \varphi)$}{~ \Return $\emph{false}$;} \EndIf
\EndFor
\State \Return $\emph{true}$;
\EndProcedure

\end{algorithmic}
\end{algorithm}

%% file: impl.tex
\section{Implementation} \label{sec:impl}

We have implemented the proposed verification technique in a new tool called \system, which consists of approximately {10,500} lines of Java code in total. \system utilizes the Z3 SMT solver~\cite{z3} to automate reasoning over the theory of relational algebra with updates. In particular, we decide the validity of a $\tra$ formula $\phi$ by asking Z3 whether the $\tra$ axioms logically imply $\phi$.  All queries to the solver are configured to have a time budget of 2 seconds, and we assume that the answer to any  query exceeding the time budget is ``invalid''.
In the remainder of this section, we describe several important optimizations that we have found necessary for making \system practical.

\input{fig-theorem}

\paragraph{Redundant axioms} During the development  of \system, we have found that many validity queries cannot be resolved due to Z3's limited capabilities for performing inductive reasoning. In particular, some of the $\tra$ theorems needed for proving equivalence require performing structural induction over lists; but Z3 times out in most of these cases. In our implementation, we address this issue by providing a redundant set of additional  axioms, which are logically implied by the $\tra$ axioms. Figure~\ref{fig:theorem}  shows a representative subset of the additional theorems that we use when issuing validity queries to Z3. Because these axioms alleviate the need for performing induction, many queries that would otherwise time out can now be successfully proven using Z3. We would like to emphasize that we came up with these redundant axioms during tool development and did not have to add any additional axioms while evaluating \system on real-world examples in our experiments.

\paragraph{Conjunctive queries} 
\begin{comment}While the full theory of relational algebra with updates is undecidable, we have identified a class of formulas for which we can come up with a decision procedure. Let us call  a query \emph{conjunctive} if it uses only  projection, selection, and equi-join and  all predicates  are conjunctions of equalities. It has been shown by Green that checking equivalence between such conjunctive queries is decidable~\cite{semiring}. Given two such conjunctive queries $Q, Q'$, Green shows that $Q, Q'$ are equivalent if their corresponding Datalog representations are the same modulo variable renaming. We borrow this observation to prove the validity of formulas of the form $\Phi \rightarrow Q_i = Q_i'$, which arise when checking sufficiency of a candidate bisimulation invariant $\Phi$. If $Q_i, Q_i'$ are conjunctive queries, then we use the schema mapping induced by $\Phi$ to rewrite the query $Q_i$ to another query $Q_i''$ such that $Q_i', Q_i''$ refer to the same schema elements.
We then use the same procedure given by Green to show that $Q_i'$ and $Q_i''$ are equivalent.
\end{comment}
While the full theory of relational algebra with updates is undecidable, we have identified a class of formulas for which we can come up with an optimization to check validity. Let us call a query \emph{conjunctive} if it uses only projection, selection, and equi-join and all predicates are conjunctions of equalities. As pointed out in prior work,  two conjunctive queries are equivalent under bag semantics if (and only if) they are syntactically isomorphic~\cite{pods99, semiring}. Inspired by their work, we optimize the validity checking of formulas of the form $\Phi \rightarrow Q_i = Q_i'$, which arise when checking sufficiency of a candidate bisimulation invariant $\Phi$. If $Q_i, Q_i'$ are conjunctive queries, we can use the schema mapping induced by $\Phi$ to rewrite the query $Q_i$ to another query $Q_i''$ such that $Q_i', Q_i''$ refer to the same schema elements. If they are syntactically the same modulo reordering of equalities, we can conclude the original formula is valid; otherwise, we check its $\tra$-validity using an SMT solver.

\paragraph{Invariant synthesis} Recall from  Algorithm~\ref{algo:verify} that our verification procedure looks for conjunctive invariants over a universe of predicates $\universe$. While these predicates are constructed from a small set of pre-defined templates, the number of possible instantiations of these templates grows quickly in the number of attributes and relations in the database schema. To prevent a blow-up in the size of the universe $\universe$,  we use the implementation of insertion transactions to rule out infeasible predicates.
For example, we only generate a predicate $\Pi_{[a_1, a_2, \ldots, a_n]}(A) = \Pi_{[b_1, b_2, \ldots, b_n]}(B)$ if there are two corresponding insertion transactions $U, U'$ such that $U$ inserts its argument $x_i$ to attribute $a_i$ of relation $A$, whereas $U'$ inserts $x_i$ into attribute $b_i$ of relation $B$. Similarly, we only generate a predicate of the form $\Pi_{[a_1, a_2, \ldots, a_n]}(A) = \Pi_{[b_1, b_2, \ldots, b_n]}(B \Join C)$ if $B$ and $C$ can be joined and there are two corresponding transactions $U, U'$ such that $U$ inserts its argument $x_i$ to  attribute $a_i$ of relation $A$, whereas $U'$ inserts the same argument into attribute $b_i$ of relation $B$ or $C$.
We have found these heuristics work quite well in that they do not lead to a loss of completeness in practice but  significantly reduce the number of predicates considered by the invariant synthesis algorithm.

\paragraph{Proving refinement.}  Recall that proving refinement requires showing that the inferred simulation invariant $\Phi$ is projectively sufficient, i.e.,
\[ \Big ( \Phi \land \bigwedge_{x_j \in \vec{x}} x_j = y_j \Big ) \models   \exists L. ~ Q_i = \Pi_L(Q_i') 
\]
In our implementation, we determine the $\tra$-validity of this formula by instantiating the existentially quantified variables $L$ with attributes in the database schema and check validity for each possible instantiation. In particular, suppose that $Q_i, Q_i'$ can contain attributes $A, A'$ respectively.  Each instantiation of $L$  essentially corresponds to a mapping $M$ such that $M(A) \subseteq A'$. Our implementation rank-orders candidate mappings based on similarity metrics between attribute names and tries more likely instantiations first.

%% file: fig-theorem.tex
\newcounter{thmct}
\newcommand{\thmct}{\stepcounter{thmct}\arabic{thmct}.}

\begin{figure}
\centering
\[
\begin{array}{r l r}
    \thmct & A \cup \nil  = A & ({\rm \cup ~ nil}) \\
    \thmct & A \Join \nil = \nil & ({\rm \Join nil}) \\
    \thmct & A - A = \nil & ({\rm - ~ nil}) \\
    \thmct & \Pi_L(A \cup B) = \Pi_L(A) \cup \Pi_L(B) & ({\rm \Pi \cup distributivity}) \\
    \thmct & \sigma_\phi(A \cup B) = \sigma_\phi(A) \cup \sigma_\phi(B) & ({\rm \sigma \cup distributivity}) \\
    \thmct & (A \cup B) \Join C = (A \Join C) \cup (B \Join C) & ({\rm \Join \cup ~ distributivity}) \\
    \thmct & \sigma_\phi(A) \Join B = \sigma_\phi(A \Join B) & ({\rm \sigma \Join associativity 1}) \\
    \thmct & A \Join \sigma_\phi(B) = \sigma_\phi(A \Join B) & ({\rm \sigma \Join associativity 2}) \\
    \thmct & \sigma_\phi(\sigma_\phi(A)) = \sigma_\phi(A) & ({\rm \sigma ~ idempotence}) \\
    \thmct & (\Pi_L(A) = \Pi_{L'}(B) \land \phi \leftrightarrow \phi'[L/L']) \to \Pi_L(\sigma_{\phi}(A)) = \Pi_{L'}(\sigma_{\phi'}(B)) & ({\rm \Pi \sigma ~ introduction}) \\
    \thmct & \Pi_L (A) = \Pi_{L'} (B) \to \Pi_L (A \awr{L_i}{v}) = \Pi_{L'} (B \awr{L'_i}{v}) & ({\rm \Pi \awr{}{} ~ introduction}) \\
\end{array}
\]
\caption{List of additional (redundant) axioms used for proving validity.}
\label{fig:theorem}
\end{figure}

%% file: eval.tex
\section{Evaluation} \label{sec:eval}

In this section, we evaluate the ideas proposed in this paper by performing two experiments. The goal of our first experiment is to investigate the practical relevance of the class of equivalence checking problems that we address in this paper. In our second experiment, we  evaluate the practicality and usefulness of the \system tool by using it to verify equivalence/refinement between different versions of database-driven web applications.

\subsection{Study of Real-World Web Applications}\label{sec:empirical}

To the best of our knowledge, there is no prior work on proving equivalence between database-driven applications, so one may wonder whether this problem is  important in practice. In other words, do  developers need to re-implement parts of a web application due to changes to the underlying database schema? While there are many articles and on-line posts that attest to the frequency and difficulty of this task, we nonetheless perform an empirical study that aims to answer precisely this question.

\input{tab-category}

In our study, we collect 100 database-driven web applications that (a) are written in Ruby-on-Rails and (b) have at least 400 commits on Github. We choose to study applications written in Ruby-on-Rails because the Rails framework provides a convenient way to view changes to the database schema. Among the applications  we have studied,  the underlying database schema had been changed \emph{at least} once in \emph{all} 100 applications, with the average application exhibiting around 28 schema changes over their lifetime. Table~\ref{tab:category} categorizes the schema changes that we have encountered, and Figure~\ref{fig:empirical} shows the percentage of applications that have undergone schema change of a certain type (addition, deletion etc.). For example, according to Figure~\ref{fig:empirical}, some attributes or relations were deleted from 62\% of these applications at least once in the commit history.

\input{barchart-empirical}

Of course, not all schema change categories listed in Table~\ref{tab:category} require significant rewriting of parts of the application code. For example,  100\% of the applications in our dataset 
exhibit a schema change induced by the addition of relations or attributes, but such additions typically require only minor changes to the application code. Therefore, many of the schema change categories listed in Table~\ref{tab:category} (e.g., addition, renaming) are not particularly interesting from  a verification perspective.  However, one of these categories, namely \emph{structure}, requires significant rewriting of the application code. In our definition, structural schema changes include splitting a relation into multiple ones, moving attributes from one relation to another, merging multiple relations into a single one etc. As shown in Figure~\ref{fig:empirical}, 44\% of the applications in our dataset have undergone a structural schema change at least once during their life time. Therefore, this empirical study demonstrates that structural schema changes are quite common and that web application developers can benefit from equivalence checking techniques that can handle non-trivial changes to the database schema.

\subsection{Evaluation of \system} \label{sec:result}
In the previous subsection, we demonstrated the importance of the problem that we introduce in this paper; however, we have not yet evaluated the feasibility of our solution.  Our goal in this section is to evaluate our \emph{solution} by using \system to prove equivalence/refinement between different versions of 21 database-driven web applications containing over 1000 transactions in total. 

\input{tab-benchmark}

\paragraph{Benchmarks} To perform this evaluation, we collect benchmarks from two different sources, namely challenging refactoring examples from textbooks/tutorials and different versions of web applications collected from Github. The textbook examples are useful for evaluating \system, as they illustrate challenging database refactoring tasks that require non-trivial changes to the application code. 
The remaining half of the benchmarks are real-world web applications taken from the dataset used in our empirical evaluation. Specifically, we evaluate \system on two different versions $A, B$ of the application such that (a) $A, B$ are consecutive versions in the commit history, (b) $B$ is obtained from $A$ by performing a structural schema change that requires rewriting parts of the application code, and (c) $B$ is meant to be equivalent to (or a refinement of) $A$. Table~\ref{tab:benchmark} describes the source of each benchmark and changes to the database schema between the two versions. Since our current implementation requires manually translating the web application to our IR representation, we only used the first 10 real-world applications that meet the afore-mentioned criteria.~\footnote{It took us approximately three days of  manual effort to translate all of these applications to our IR. We plan to automate this transformation in future work. }

\input{tab-result}

\paragraph{Results} 
Table~\ref{tab:result} summarizes the  results of our evaluation of \system on these benchmarks. For each benchmark, the column labeled \emph{Type} shows whether we used \system to check refinement ($\preceq$) or equivalence ($\simeq$), and \emph{Trans} shows the number of transactions in each application. The next two columns provide information about the number of relations and total number of attributes in the source and target schema, respectively. The last four columns provide information about \system results: \emph{Status} shows whether \system was able to verify the desired property (i.e., equivalence or refinement), and \emph{Time} provides total running time in seconds. The column labeled \emph{Iters} shows the number of iterations that \system takes to find an inductive simulation (or bisimulation) invariant. Finally, the last column labeled \emph{Queries} shows the number of validity modulo $\tra$ checks issued by \system.

As we can see from Table~\ref{tab:result}, \system is able to successfully verify the desired property for {20 out of 21} benchmarks. The running time of the tool ranges between {0.2 seconds} for small textbook examples with a few transactions to 150 seconds for large, real-world benchmarks with hundreds of transactions. As expected, the running time of \system on real-world benchmarks is typically much longer (46.8 seconds on average) than on textbook examples (11.3 seconds on average). However, some textbook examples (namely benchmarks 5 and 10) take longer than  some of the real-world examples because  many iterations are required to find an inductive bisimulation invariant. As shown in Figure~\ref{fig:correlation}(b), the running time of \system is roughly linear in the number of validity queries to the SMT solver. Because the number of validity checks depends on the number of transactions in the application as well as the number of iterations required for finding an inductive bisimulation invariant, Figure~\ref{fig:correlation}(a) also shows that system's running time is roughly linear with respect to $\#\emph{transactions} \times \# \emph{iterations}$.

\input{scatter-correlation}

\paragraph{Cause of false positives.} 
As shown in Table~\ref{tab:result}, \system fails to verify equivalence for benchmark 20, where the schema change involves moving the shared attributes of two relations into a new polymorphic relation. Upon further inspection, we determined this warning to be a false positive that is caused by a shortcoming of our inference engine for synthesizing bisimulation invariants. In particular, proving equivalence of this benchmark requires a bisimulation invariant of the form $\Pi_L(R) = \Pi_{L'}(\sigma_{\phi}(R'))$, which is currently not supported in our implementation (recall Section~\ref{sec:inv}). While it is possible to 
extend our templates to include predicates of this form, this modification would significantly increase the search space, as the selection predicate $\phi$ can be instantiated in \emph{many} different ways.

%% file: tab-category.tex
\begin{table}[t]
\centering
\newcolumntype{A}{ >{\arraybackslash} m{0.10\textwidth} }
\newcolumntype{B}{ >{\arraybackslash} m{0.80\textwidth} }
\newcolumntype{C}{ >{\centering\arraybackslash} m{0.06\textwidth} }
\begin{tabular}{|A|B|}
\hline
Addition & Introduce new relations and/or attributes. \\
\hline
Deletion & Remove existing relations and/or attributes. \\
\hline
Type & Change the type of attributes in the schema. \\
\hline
Rename & Rename existing relations and/or attributes in the schema. \\
\hline
Structure & Change the structure of the schema, including but not limited to moving attributes from one relation to another, splitting a relation into multiple ones, merging multiple relations into one. \\
\hline
Index & Add, delete, or modify indices of the database system. \\
\hline
Constraint & Add, delete, or modify constraints of the database system. \\
\hline
Other & Bulk schema change of several other categories. \\
\hline
\end{tabular}
\vspace{0.10in}
\caption{Summary of categorization}
\label{tab:category}
\end{table}

%% file: barchart-empirical.tex
\definecolor{bblue}{HTML}{0064FF}
\definecolor{rred}{HTML}{C0504D}
\definecolor{ggreen}{HTML}{9BBB59}
\definecolor{ppurple}{HTML}{9F4C7C}

\begin{figure}[!t]
\centering
\begin{tikzpicture}
    \begin{axis}[
        width  = 10cm,
        height = 5cm,
        y=0.04cm,
        x=1.35cm,
        major x tick style = transparent,
        ybar=2*\pgflinewidth,
        bar width=16pt,
        ymajorgrids = true,
        ylabel = {Applications (\%)},
        ylabel style={yshift=-3mm},
        symbolic x coords={Addition,Deletion,Type,Rename,Structure,Index,Constraint,Other},
        xtick = data,
        ytick = {0,20,40,60,80,100},
        scaled y ticks = false,
        enlarge x limits=0.10,
        ymin = 0,
        ymax = 101,
        legend cell align=left,
        legend style={
                at={(0.5,1.08)},
                legend columns=-1,
                anchor=north,
        }
    ]
    \hspace*{-2mm}
        \addplot[style={bblue,fill=bblue,mark=none}]
        coordinates {(Addition,100)(Deletion,62)(Type,61)(Rename,54)(Structure,44)(Index,42)(Constraint,42)(Other,68)};

    \end{axis}
\end{tikzpicture}
\caption{Empirical Study. Percentage of applications that have different categories of database schema changes.}
\label{fig:empirical}
\end{figure}
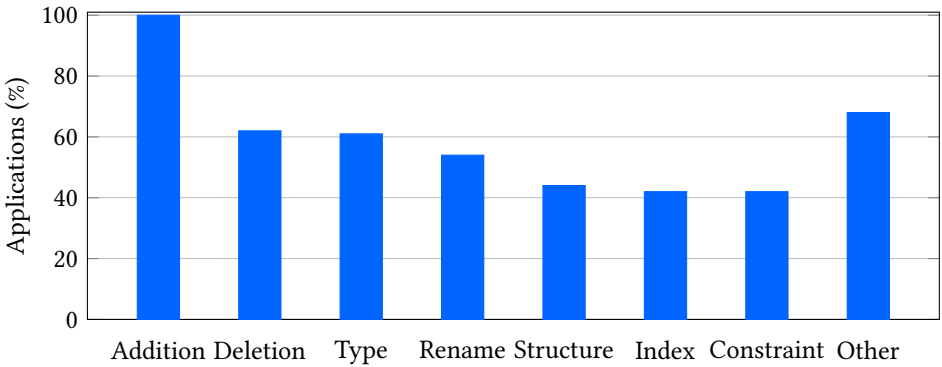

%% file: tab-benchmark.tex
\begin{table}[t]
\centering
\newcolumntype{A}{ >{\arraybackslash} m{0.23\textwidth} }
\newcolumntype{B}{ >{\arraybackslash} m{0.60\textwidth} }
\newcolumntype{C}{ >{\centering\arraybackslash} m{0.04\textwidth} }
\begin{tabular}{|c|c|A|B|}
\hline
& {\bf ID} & {\bf Source} & {\bf Description} \\
\hline
\hline
\multirow{11}{*}{\begin{turn} {90}\makecell{\small \bf textbook \ bench}\end{turn}}
& 1 & Oracle tutorial & Merge relations \\
\cline{2-4}
& 2 & Oracle tutorial & Split relations \\
\cline{2-4}
& 3 & Textbook & Split relations \\
\cline{2-4}
& 4 & Textbook & Merge relations \\
\cline{2-4}
& 5 & Textbook & Move attributes \\
\cline{2-4}
& 6 & Textbook & Rename attributes \\
\cline{2-4}
& 7 & Textbook & Introduce associative relations \\
\cline{2-4}
& 8 & Textbook & Replace the surrogate key with natural key \\
\cline{2-4}
& 9 & Textbook & Introduce new attributes \\
\cline{2-4}
& 10 & Textbook & Denormalization \\
\hline
\hline
\multirow{10}{*}{\begin{turn} {90}\makecell{\small \bf real-world \ bench}\end{turn}}
& 11 & cdx & Rename attributes and split relations \\
\cline{2-4}
& 12 & coachup & Split relations \\
\cline{2-4}
& 13 & 2030Club & Split relations \\
\cline{2-4}
& 14 & rails-ecomm & Split relations and introduce new attributes \\
\cline{2-4}
& 15 & royk & Introduce and move attributes \\
\cline{2-4}
& 16 & MathHotSpot & Rename relations and move attributes \\
\cline{2-4}
& 17 & gallery & Split relations \\
\cline{2-4}
& 18 & DeeJBase & Rename attributes and split relations \\
\cline{2-4}
& 19 & visible-closet-1 & Split relations \\
\cline{2-4}
& 20 & visible-closet-2 & Move attributes to a polymorphic relation \\
\cline{2-4}
& 21 & probable-engine & Merge relations \\
\hline
\end{tabular}
\vspace{0.10in}
\caption{Benchmark source and description. The textbook refers to~\cite{refactordb} and the Oracle tutorial refers to~\cite{oracle}.}
\label{tab:benchmark}
\end{table}

%% file: tab-result.tex
\begin{table}[t]
\centering
\newcolumntype{A}{ >{\arraybackslash} m{0.30\textwidth} }
\newcolumntype{B}{ >{\centering\arraybackslash} m{0.07\textwidth} }
\newcolumntype{C}{ >{\centering\arraybackslash} m{0.04\textwidth} }
\newcolumntype{D}{ >{\centering\arraybackslash} m{0.05\textwidth} }
\begin{tabular}{|c|c|C|c|B|B|B|B|D|c|c|B|}
\hline
& \multirow{2}{*}{{\bf ID}} & \multirow{2}{*}{\hspace{-0.05in}{\bf Type}} & \multirow{2}{*}{{\bf Trans}} & \multicolumn{2}{c|}{{\bf Source Schema}} & \multicolumn{2}{c|}{{\bf Target Schema}} & \multirow{2}{*}{\hspace{-0.05in}{\bf Status}} & {\bf Time} & \multirow{2}{*}{\bf Iters} & \multirow{2}{*}{\hspace{-0.05in}{\bf Queries}} \\
\cline{5-8}
& & & & {\bf Rels} & {\bf Attrs} & {\bf Rels} & {\bf Attrs} & & {\bf (s)} & & \\
\hline
\hline
\multirow{11}{*}{\begin{turn} {90}\makecell{\small \bf textbook \ bench}\end{turn}}
& 1 & $\simeq$ & 4 & 2 & 8 & 1 & 6 & $\checkmark$ & 2.4 & 2 & 10 \\
\cline{2-12}
& 2 & $\simeq$ & 19 & 3 & 17 & 7 & 25 & $\checkmark$ & 11.5 & 5 & 188 \\
\cline{2-12}
& 3 & $\simeq$ & 10 & 1 & 6 & 2 & 7 & $\checkmark$ & 2.5 & 2 & 32 \\
\cline{2-12}
& 4 & $\simeq$ & 10 & 2 & 7 & 1 & 6 & $\checkmark$ & 0.3 & 1 & 16 \\
\cline{2-12}
& 5 & $\simeq$ & 7 & 2 & 5 & 2 & 5 & $\checkmark$ & 34.1 & 7 & 147 \\
\cline{2-12}
& 6 & $\simeq$ & 5 & 1 & 2 & 1 & 2 & $\checkmark$ & 0.2 & 1 & 5 \\
\cline{2-12}
& 7 & $\simeq$ & 8 & 2 & 5 & 3 & 6 & $\checkmark$ & 3.1 & 2 & 61 \\
\cline{2-12}
& 8 & $\simeq$ & 10 & 2 & 9 & 2 & 8 & $\checkmark$ & 0.4 & 1 & 18 \\
\cline{2-12}
& 9 & $\preceq$ & 8 & 2 & 7 & 2 & 8 & $\checkmark$ & 0.3 & 1 & 14 \\
\cline{2-12}
& 10 & $\simeq$ & 14 & 3 & 10 & 3 & 13 & $\checkmark$ & 57.7 & 23 & 374 \\
\hline
\hline
\multirow{10}{*}{\begin{turn} {90}\makecell{\small \bf real-world \ bench}\end{turn}}
& 11 & $\simeq$ & 138 & 16 & 125 & 17 & 131 & $\checkmark$ & 90.8 & 13 & 4840 \\
\cline{2-12}
& 12 & $\simeq$ & 45 & 4 & 51 & 5 & 55 & $\checkmark$ & 23.2 & 7 & 489 \\
\cline{2-12}
& 13 & $\simeq$ & 125 & 15 & 155 & 16 & 159 & $\checkmark$ & 42.6 & 8 & 2403 \\
\cline{2-12}
& 14 & $\preceq$ & 65 & 8 & 69 & 9 & 75 & $\checkmark$ & 23.4 & 7 & 1059 \\
\cline{2-12}
& 15 & $\preceq$ & 151 & 19 & 152 & 19 & 155 & $\checkmark$ & 19.1 & 1 & 1307 \\
\cline{2-12}
& 16 & $\simeq$ & 54 & 7 & 38 & 8 & 42 & $\checkmark$ & 20.9 & 6 & 701 \\
\cline{2-12}
& 17 & $\simeq$ & 58 & 7 & 52 & 8 & 57 & $\checkmark$ & 54.5 & 13 & 1512 \\
\cline{2-12}
& 18 & $\simeq$ & 70 & 10 & 92 & 11 & 97 & $\checkmark$ & 28.7 & 6 & 1228 \\
\cline{2-12}
& 19 & $\simeq$ & 263 & 26 & 248 & 27 & 252 & $\checkmark$ & 150.6 & 12 & 9072 \\
\cline{2-12}
& 20 & $\simeq$ & 267 & 28 & 262 & 29 & 261 & $\times$ & - & - & - \\
\cline{2-12}
& 21 & $\simeq$ & 85 & 12 & 83 & 11 & 78 & $\checkmark$ & 13.7 & 3 & 823 \\
\hline
\end{tabular}
\vspace{0.10in}
\caption{Summary of experimental results. All  experiments are performed on a computer with Intel Xeon(R) E5-1620 v3 CPU and 32GB of memory, running Ubuntu 14.04 operating system.}
\label{tab:result}
\end{table}

%% file: scatter-correlation.tex
\definecolor{bblue}{HTML}{0064FF}
\definecolor{rred}{HTML}{C0504D}
\definecolor{ggreen}{HTML}{9BBB59}
\definecolor{ppurple}{HTML}{9F4C7C}

\begin{figure}[!t]
\centering
\begin{subfigure}[b]{0.46\textwidth}
\scalebox{0.7}{
\begin{tikzpicture}
\begin{axis}[
    domain=0:3500,
    xtick={0,1000,...,3000},
    ytick={0,30,...,180},
    xmajorgrids=true,ymajorgrids=true,
    scaled x ticks=false,
    scaled y ticks=false,
    xlabel = {Transaction Number $\times$ Iteration Number},
    ylabel = {Running Time (s)},
]
\addplot+[mark=none, line width=2pt]{0.04592*x+7.920};
\addplot+[only marks, mark=*, mark color=bblue, mark options={scale=1.5, fill=bblue}, text mark as node=true] coordinates {
    (8, 2.4)
    (95, 11.5)
    (10, 2.5)
    (10, 0.3)
    (49, 34.1)
    (5, 0.2)
    (16, 3.1)
    (10, 0.4)
    (8, 0.3)
    (322, 57.7)
    (1794, 90.8)
    (315, 23.2)
    (1000, 42.6)
    (455, 23.4)
    (151, 19.1)
    (348, 20.9)
    (754, 54.5)
    (420, 28.7)
    (3156, 150.6)
    (255, 13.7)
};
\end{axis}
\end{tikzpicture}
}
\caption{Running time with respect to number of transactions and iterations}
\end{subfigure}
\quad
\begin{subfigure}[b]{0.46\textwidth}
\scalebox{0.7}{
\begin{tikzpicture}
\begin{axis}[
    domain=0:10000,
    xtick={0,2000,...,10000},
    ytick={0,30,...,180},
    xmajorgrids=true,ymajorgrids=true,
    scaled x ticks=false,
    scaled y ticks=false,
    xlabel = {Query Number},
    ylabel = {Running Time (s)},
]
\addplot+[mark=none, line width=2pt]{0.01593*x+9.659};
\addplot+[only marks, mark=*, mark color=bblue, mark options={scale=1.5, fill=bblue}, text mark as node=true] coordinates {
    (10, 2.4)
    (188, 11.5)
    (18, 2.5)
    (16, 0.3)
    (147, 34.1)
    (5, 0.2)
    (61, 3.1)
    (18, 0.4)
    (14, 0.3)
    (374, 57.7)
    (4840, 90.8)
    (489, 23.2)
    (2403, 42.6)
    (1059, 23.4)
    (1307, 19.1)
    (701, 20.9)
    (1512, 54.5)
    (1228, 28.7)
    (9072, 150.6)
    (823, 13.7)
};
\end{axis}
\end{tikzpicture}
}
\caption{Running time with respect to number of $\tra$ validity queries}
\end{subfigure}

\caption{Relationship between running time, transaction number, iteration number and query number}
\label{fig:correlation}
\end{figure}
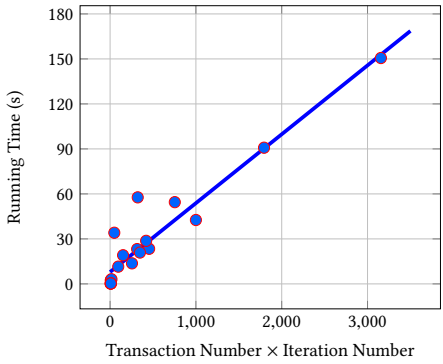
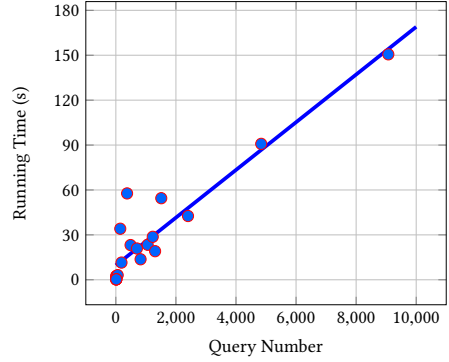

%% file: related.tex
\section{Related Work}

The research problem that we address in this paper is related to a long line of work from the programming languages and databases communities. In what follows, we survey papers that are most relevant and explain how they differ from our approach.

\vspace{0.1in}\noindent
{\bf \emph{Translation Validation.}} One of the most well-known applications of equivalence checking is \emph{translation validation}, where the goal is to prove that the compiled version of the code is equivalent to the original one.
\cite{tv-pnueli, tv-necula, tv-rinard, tv-lerner,tv-zuck,covac}. More recent work extends translation validation to \emph{parameterized equivalence checking} (PEC), which aims to prove equivalence between templatized programs representing many concrete programs~\cite{pec}. Most of the work in this area focuses on imperative programs and proves equivalence by inferring some form of bisimulation relation. Another common technique for proving equivalence is to generate a so-called \emph{product program}~\cite{covac,product-barthe} and reduce the equivalence checking problem to the safety verification of a single program. Rather than validating the correctness of compiler optimizations, our goal in this work is to show equivalence between database-driven programs before and after changes to the database schema. Our bisimulation invariants relate database states rather than program variables and are expressed in the theory of relational algebra with updates instead of standard first-order theories directly supported by SMT solvers.

\vspace{0.1in}\noindent
{\bf \emph{Contextual Equivalence.}}
There has also been a significant body of work on verifying \emph{contextual equivalence}, where the goal is to determine whether two expressions  are equivalent under any program context. One important application of contextual equivalence is to identify compiler optimization opportunities, particularly in the context of functional programming languages. For example, Sumii and Pierce define an untyped call-by-value lambda calculus with sealing and present a bisimulation-based approach to prove contextual equivalence~\cite{sumii-popl04}. They later present another sound and complete proof methodology based on bisimulation relations, but apply it to a lambda calculus with full universal, existential, and recursive types~\cite{sumii-popl05}. Koutavas and Wand extend this line of work to prove contextual equivalence in an untyped lambda calculus with an explicit store by introducing a new notion of bisimulation.  Their method enables constructive proofs in the presence of higher-order functions~\cite{koutavas-popl06}. They also extend the same proof technique to the imperative untyped object calculus~\cite{koutavas-esop06}. Sangiorgi et al. step further and develop a notion of \emph{environmental bisimulation} for higher-order languages. Their technique does not require induction on evaluation derivations  and is applicable to different calculi, ranging from pure lambda calculus to higher-order $\pi$-calculus~\cite{sangiorgi-lics07, sangiorgi-toplas11}. Existing techniques for proving contextual equivalence offer a limited degree of automation and do not address database-driven applications.

\vspace{0.1in}\noindent
{\bf \emph{Regression Verification.}} Regression verification is concerned with checking equivalence across different versions of an evolving program~\cite{regression1,symdiff1,symdiff3,felsing-ase14}.
Unlike translation validation where the two programs are expressed at different levels of abstraction and transformations are intra-procedural, regression verification deals with two programs across a refactoring, feature addition or bug-fix with changes possibly spread across procedures. 
Approaches range from the use of differential symbolic execution~\cite{symdiff3} for loop-free and recursion-free programs, techniques based on uninterpreted functions to deal with mutual recursion~\cite{regression1,symdiff1}, use of  mutual summaries for checking heap-manipulating programs~\cite{symdiff4,wood-esop17}, to the use of standard invariant inference for inferring intermediate relational invariants~\cite{symdiff2,felsing-ase14}.
Because the invariants and  assertions relate program variables, the resulting verification conditions are expressible in first-order theories  supported by  SMT solvers. Existing regression verification techniques do not handle database-driven applications and therefore cannot reason about changes to the database schema.

\vspace{0.1in}\noindent
{\bf \emph{Relational Program Logics.}} The problem of verifying equivalence between a pair of programs can be viewed as a special kind of \emph{relational verification} problem in which the goal is to prove relational Hoare triples of the form $\{ P \} \ S_1 \sim S_2 \  \{Q\}$. Here, $P$ is a relational pre-condition that relates inputs to programs $S_1, S_2$, and $Q$ is a relational post-condition that relates their outputs. In the context of equivalence checking, the pre-condition simply assumes equality between program inputs, and the post-condition asserts equality between their outputs. Prior work has presented program logics, such as Relational Hoare Logic and Cartesian Hoare logic, for showing such relational correctness properties~\cite{rhl,chl,rsl}.  As mentioned earlier, another common technique for proving relational correctness properties is to construct a product program $S_1 \times S_2$, which is semantically equivalent to $S_1; S_2$ but that is somehow easier to verify~\cite{product-barthe,product-barthe2}. In contrast to existing relational verification techniques that work on imperative programs, our approach addresses database-driven programs that work over different schema. Furthermore, while the main focus of this paper is to verify equivalence/refinement between programs, we believe our technique can be easily extended for proving other relational correctness properties.

\vspace{0.1in}\noindent
{\bf \emph{Testing and Verification of Database-Driven Applications.}} There has been a significant body of prior work on testing and analyzing database-driven applications~\cite{wave1, wave2, wave3,rubicon,dpf,apollo,agenda1,wassermann,emmi, BGL98}.
Some of these techniques aim to verify integrity constraints, find shallow bugs, or identify security vulnerabilities, while others attempt to uncover violations of functional correctness properties. For example, \cite{BGL98} statically verifies that database transactions preserve integrity constraints by means of computing weakest preconditions.
As another example, the Agenda framework generates test cases by randomly populating the database with tuples that satisfy schema constraints~\cite{agenda1,agenda2}, and several papers use concolic testing to find crashes or SQL injection vulnerabilities~\cite{apollo,wassermann}. Gligoric and Majumdar describe an explicit state model checking technique for database-driven applications and use this technique to find concurrency bugs~\cite{dpf}.

There have also been proposals for checking functional correctness of database-driven applications. For example, Near and Jackson present a bounded verifier that uses symbolic execution to check functional correctness properties specified using an extension of the RSpec specification language~\cite{rubicon}. As another example, the WAVE project allows users to specify functional correctness properties using LTL formulas and model checks a given database-driven application against this specification~\cite{wave1,wave2,wave3}. The WAVE framework can also be used to automatically synthesize web applications from declarative specifications~\cite{wave2}. However, we are not aware of any existing work on verifying relational correctness properties of database-driven applications.

\vspace{0.1in}\noindent
{\bf \emph{Transformation of Database-Driven Applications.}}
There has also been some prior  work on optimizing database-driven applications~\cite{qbs,fiat}. For example, the \qbs system detects performance bugs in web applications and repairs them by replacing inefficient Java code with SQL queries~\cite{qbs}. Similar to our work, \qbs also models database relations using lists and defines a \emph{theory of ordered relations} for reasoning about loop invariants. However, the theory of ordered relations proposed by Cheung et al. does not allow reasoning about updates to database relations.

Another related work in this space is the \fiat system for synthesizing abstract datatypes (ADTs) with SQL-like operations~\cite{fiat}. Given a user-provided specification, \fiat uses domain-specific refinement theorems and automation tactics to transform the given specification into an (ideally efficient) implementation. However, unlike our approach, \fiat requires  different implementations to use the same schema.

\vspace{0.1in}\noindent
{\bf \emph{Query Equivalence.}}
Query equivalence has been a long-standing problem in the database community. For example, Chandra and Merlin study the equivalence of conjunctive queries under set semantics and show that every conjunctive query has a unique representation up to isomorphism~\cite{conjquery}.  Aho et al. use tableaus to represent conjunctive queries and present a polynomial time algorithm for deciding equivalence between certain kinds of conjunctive queries~\cite{sagiv79}. Sagiv and Yannanakis generalize this tableau approach to union and difference operators, but place limitations on the use of  difference~\cite{sagiv80}. In more recent work, Green studies equivalence between conjunctive queries and gives an algorithm for deciding their equivalence~\cite{semiring}. As mentioned in Section~\ref{sec:impl}, our implementation leverages insights from this work when checking validity of certain classes of $\tra$ formulas.

In the past year, there has been significant activity on proving query equivalence using interactive theorem provers and constraint solvers. Specifically,  Chu et al. define a new semantics for SQL based on K-Relations and homotopy type theory~ and develop a Coq library to interactively prove equivalence between SQL queries~\cite{hottsql}. Another recent work by Chu et al. provides a greater degree of automation by incorporating constraint solvers~\cite{cosette}. Specifically, their tool, \cosette,  first translates each SQL query into a corresponding logical formula and uses an SMT solver to find a counterexample that shows the queries are \emph{not} equivalent. If \cosette fails to find a counterexample, it then tries to prove equivalence using the Coq theorem prover augmented with a library of domain-specific tactics. While tools like \cosette could be useful for deciding validity of some of our $\tra$ formulas, existing tools do not support reasoning about updates to the database.

\vspace{0.1in}\noindent
{\bf \emph{Schema Evolution.}}
In the database community, \emph{schema evolution} refers to the problem of evolving database schemas to adapt to requirement changes. Schema evolution also typically requires  data migration and synchronization of legacy transactions, and many papers from the database community address the problem of facilitating schema evolution~\cite{evo-survey1, evo-survey2, prism, schema-matching, cst}. For example, the {\sc Prism} and {\sc Prism++} projects~\cite{prism} investigate techniques for automatically evolving transactions using a given set of SQL-based schema modification operators.

As another example, Fagin et al. study the schema mapping adaptation problem using composition and inversion operators~\cite{schema-matching} and work by Visser utilizes strategy combinators and point-free program calculation to handle coupled transformation of schemas, data instances, queries, and constraints~\cite{cst}.
In contrast to our method, existing work on schema evolution cannot be used to prove equivalence between any pair of database-driven applications.

\vspace{0.1in}\noindent
{\bf \emph{Schema Equivalence.}}
There has also been some work on proving schema equivalence under various different definitions of equivalence~\cite{beeri-stoc79,miller-vldb93,rosenthal-tods94}. Under one definition, two schemas are considered equivalent if there is a bijection from the set of database instances from one schema to that of another~\cite{miller-vldb93, rosenthal-tods94}. According to another definition,  two schemas $S_1, S_2$ are equivalent if there is a  query mapping from $S_1$ to $S_2$ such that its inverse is also a valid mapping from $S_2$ to $S_1$\cite{atzeni-tcs82, hull-siam86, albert-jcss99}. While many of these papers provide algorithms for deciding schema equivalence according to these definitions, they do not address the problem of verifying equivalence between  applications that operate over databases with different schemas.

%% file: limitation.tex
\section{Limitations}\label{sec:limitation}

As a research prototype, our current implementation of \system has a number of limitations:
First, \system analyzes programs that are written in the database transaction language given in Figure~\ref{fig:ir}. Hence, the applicability of \system relies on translating the original database-driven application to our IR, which abstracts programs as a fixed set of queries and updates to the database. Thus, programs that use dynamically generated transactions or control-flow constructs cannot be translated into our IR. Second, \system synthesizes simulation and bisimulation invariants by finding the strongest conjunctive formula over a given class of predicates. However, as exemplified by the false positive from our evaluation, \system may not be able to prove equivalence if the bisimulation requires additional predicates (or boolean connectives) beyond the ones we consider.   Third, \system axiomatizes $\tra$ using the theory of lists, which is also undecidable. Therefore, the SMT solver may time-out when checking validity queries over the theory of lists. Fourth, \system can only be used to prove equivalence but not \emph{disequivalence}. In particular, \system cannot provide witnesses to prove that two applications are indeed not equivalent. Finally, our verification technique proves equivalence under \emph{list semantics}. Thus, if a web application uses set/bag semantics to represent the results of database queries, \system may end up reporting false positives. However, despite these limitations, our evaluation shows that \system is still practical and that it can verify equivalence between different versions of many real-world web applications.

%% file: concl.tex
\section{Conclusions and Future Work}

In this paper, we have introduced the problem of verifying equivalence between applications that operate over databases with different schemas.  As shown by our empirical study from Section~\ref{sec:empirical}, this problem is practically relevant because many web applications undergo a structural schema change that requires re-implementation of parts of the application code. Our technique for verifying equivalence relies on finding an inductive bisimulation invariant over the theory of relational algebra with updates and automatically discharges verification conditions that are needed to prove equivalence. We have implemented this technique in a new tool called \system and evaluate it on different versions of database-driven web applications containing up to hundreds of transactions. Our evaluation shows that  \system can successfully verify equivalence/refinement between 95\% of the web applications it was evaluated on and that involve structural changes to the underlying database schema.

We see this work as a first step towards verifying equivalence of database-driven applications. In future work, we plan to address  the limitations discussed in Section~\ref{sec:limitation} to increase the applicability of the tool. First, we plan to develop a richer intermediate representation that can model {conditional updates and other} interactions between the application and the database. While this extension would likely not necessitate fundamental changes to the proposed verification methodology, it would require performing reasoning over a richer logical theory (e.g., combined theory of Presburger arithmetic, relational algebra, arrays, etc). Second,  we plan to automate the translation of web applications written in realistic languages (e.g., Ruby, PHP, etc) into our intermediate language. Third, we are interested in adding counterexample-generation capabilities to \system so that it can prove that two applications are definitely not equivalent.

%% file: ack.tex
\section*{Acknowledgments}\label{sec:ack}
We would like to thank James Cheney, Kostas Ferles, Xinyu Wang, Yu Feng, Jacob Van Geffen, and the anonymous POPL'18 reviewers for their thorough and helpful comments on an earlier version of this paper. This material is based on research sponsored by NSF Awards \#1712060 and \#1453386.

%% file: proof.tex
\section{Proof of Theorems} \label{sec:proof}

%%%%%%%%%%%%%%%%%%%%%%%%%%%%%%%%%%%%%%%%%%%%%%%%%%%%%%%%%%%%%%%%
% sp predicate
%%%%%%%%%%%%%%%%%%%%%%%%%%%%%%%%%%%%%%%%%%%%%%%%%%%%%%%%%%%%%%%%
\begin{lemma} \label{thm:sp-pred}
    Suppose $\denot{\phi}_{\sigma, \Delta, x} = p$, then $\attnum(\Delta) \uplus \sigma \uplus [x \mapsto c] \models \phi(x)$ if and only if $p[c/x] = true$.
\end{lemma}

\begin{proof}
By structural induction on $\phi$.
\begin{enumerate}
    \item Base case: $\phi = a_1 \odot a_2$. We know $\denot{\phi}_{\sigma, \Delta, x} = \emph{lookup}(x, a_1) \odot \emph{lookup}(x, a_2)$ and $\phi(x) = get(i, x) \odot get(j, x)$ where $i,j$ are the indices of $a_1, a_2$. By the assumption that $\attnum$ is correct, we have $\emph{lookup}(x, a) = get(k, x)$ if $k$ is the index of $a$. Thus $\emph{lookup}(c, a_1) \odot \emph{lookup}(c, a_2) \leftrightarrow get(i, c) \odot get(j, c)$.
    \item Base case: $\phi = a \odot v$. We know $\denot{\phi}_{\sigma, \Delta,x} = \emph{lookup}(x, a_1) \odot v[\sigma]$ and $\phi(x) = get(i, x) \odot v$. Again, by the assumption that $\attnum$ is correct, we have $\sigma \models get(i, c) \odot v$ if and only if $\emph{lookup}(c, a) \models v[\sigma]$ for any value of $c$.
    \item Base case: $\phi = a \in Q$. Suppose $i$ is the index of $a$, and $\denot{Q}_{\sigma, \Delta} = R$. We have
        \[
            \denot{\phi}_{\sigma, \Delta, x} = \emph{contains}(\emph{lookup}(x, a), map(R, \lambda y.~\emph{head}(\emph{vals}(y))))
        \]
        and
        \[
            \phi(x) = \emph{contains}(get(i, x), \emph{firstColumn}(\attnum(Q)))
        \]
        where $\emph{firstColumn}(\attnum(Q))$ means extracting the first value of each ``row'' (inner list) in $\attnum(Q)$ and converting it to a list, which follows from the definition of $(a_i \in t)(h)$ in Figure~\ref{fig:axiom-aux}. By the correctness of $\attnum$, $\emph{lookup}(x, a) = get(i, x)$ for the same value of $x$. Also, from Lemma~\ref{thm:sp-query} and $\denot{Q}_{\sigma, \Delta} = R$, we have $\attnum(\Delta) \uplus \sigma \models \attnum(Q) = \attnum(R)$. Thus,
        \[
            \emph{firstColumn}(\attnum(Q)) = map(R, \lambda y.~\emph{head}(\emph{vals}(y)))
        \]
        Hence, $\attnum(\Delta) \uplus \sigma \uplus [x \mapsto c] \models \emph{contains}(get(i, x), \emph{firstColumn}(\attnum(Q)))$ if and only if
        \[
            \emph{contains}(\emph{lookup}(c,a), map(R, \lambda y. \emph{head}(\emph{vals}(y)))
        \]
        evaluates to $true$.
    \item Inductive case: $\phi = \phi_1 \land \phi_2$. Suppose $\denot{\phi_1}_{\sigma, \Delta, x} = p_1$ and $\attnum(\Delta) \uplus \sigma \uplus [x \mapsto c] \models \phi_1(x)$ iff $p_1[c/x] = true$. Also suppose $\denot{\phi_2}_{\sigma, \Delta, x} = p_2$ and $\attnum(\Delta) \uplus \sigma \uplus [x \mapsto c] \models \phi_2(x)$ iff $p_2[c/x] = true$. Obviously, $(p_1 \land p_2)[c/x] = true$ iff both $p_1[c/x]$ and $p_2[c/x]$ are $true$. Thus $\attnum(\Delta) \uplus \sigma \uplus [x \mapsto c] \models (\phi_1 \land \phi_2)(x)$ iff $(p_1 \land p_2)[c/x] = true$.
    \item Inductive case: $\phi = \phi_1 \lor \phi_2$. Suppose $\denot{\phi_1}_{\sigma, \Delta, x} = p_1$ and $\attnum(\Delta) \uplus \sigma \uplus [x \mapsto c] \models \phi_1(x)$ iff $p_1[c/x] = true$. Also suppose $\denot{\phi_2}_{\sigma, \Delta, x} = p_2$ and $\attnum(\Delta) \uplus \sigma \uplus [x \mapsto c] \models \phi_2(x)$ iff $p_2[c/x] = true$. Obviously, $(p_1 \lor p_2)[c/x] = true$ iff either $p_1[c/x]$ or $p_2[c/x]$ is $true$. Thus $\attnum(\Delta) \uplus \sigma \uplus [x \mapsto c] \models (\phi_1 \lor \phi_2)(x)$ iff $(p_1 \lor p_2)[c/x] = true$.
    \item Inductive case: $\phi = \neg \phi_1$. Suppose $\denot{\phi_1}_{\sigma, \Delta, x} = p_1$ and $\attnum(\Delta) \uplus \sigma \uplus [x \mapsto c] \models \phi_1(x)$ iff $p_1[c/x] = true$. Since $(\neg p_1)[c/x] = true$ iff $p_1[c/x]$ is $false$, $\attnum(\Delta) \uplus \sigma \uplus [x \mapsto c] \models (\neg \phi_1)(x)$ iff $(\neg p_1)[c/x] = true$.
\end{enumerate}
\end{proof}

%%%%%%%%%%%%%%%%%%%%%%%%%%%%%%%%%%%%%%%%%%%%%%%%%%%%%%%%%%%%%%%%
% sp query
%%%%%%%%%%%%%%%%%%%%%%%%%%%%%%%%%%%%%%%%%%%%%%%%%%%%%%%%%%%%%%%%
\begin{lemma} \label{thm:sp-query}
If $\denot{Q}_{\sigma, \Delta} = R$, then $\attnum(\Delta) \uplus \sigma \models \attnum(Q) = \attnum(R)$.
\end{lemma}

\begin{proof}[Proof Sketch]
By structural induction on $Q$.
\begin{enumerate}
    \item Base case: $Q = R$. $\denot{Q}_{\sigma, \Delta} = \Delta(R)$. Thus, by the correctness of $\attnum$, $\attnum(\Delta) \models \attnum(Q) = \attnum(R)$.
    \item Inductive case: $Q = \Pi_\psi(Q_1)$. Suppose that $\denot{Q_1}_{\sigma, \Delta} = R_1$, $\attnum(\Delta) \uplus \sigma \models \attnum(Q_1) = \attnum(R_1)$, and $\denot{\Pi_\psi(Q_1)}_{\sigma, \Delta} = R$, where $R = \emph{map}(R_1, \lambda x. \emph{filter}(x, \lambda y. \emph{contains}(\emph{first}(y), \psi)))$. From the definition of $\Pi'_l(h)$ and $\Pi_l(t)$ in Figure~\ref{fig:axiom}, we know $\Pi'_l(h) = \emph{filter}(h, ~ \lambda z. \emph{contains}(z, l))$ and $\Pi_l(t) = \emph{map}(t, \lambda x. \emph{filter}(x, \lambda z. \emph{contains}(z, l)))$. By the correctness of $\attnum$, we have
    \[
    \begin{array}{r c l}
        \attnum(R) &=& \emph{map}(\attnum(R_1), ~ \lambda x. \emph{filter}(x, \lambda z. \emph{contains}(z, \attnum(\psi)))) \\
                   &=& \emph{map}(\attnum(Q_1), ~ \lambda x. \emph{filter}(x, \lambda z. \emph{contains}(z, \attnum(\psi)))) \\
                   &=& \attnum(Q)
    \end{array}
    \]
    Thus, $\attnum(\Delta) \uplus \sigma \models \attnum(Q) = \attnum(R)$.
    \item Inductive case: $Q = \sigma_\phi(Q_1)$. Suppose that $\denot{Q_1}_{\sigma, \Delta} = R_1$, $\attnum(\Delta) \uplus \sigma \models \attnum(Q_1) = \attnum(R_1)$, $\denot{\phi}_{\sigma, \Delta, x} = p$, and $\denot{\sigma_\phi(Q_1)}_{\sigma, \Delta} = R$, where $R = \emph{filter}(R_1, \lambda x. p)$. From Lemma~\ref{thm:sp-pred}, we know $\phi$ and $p$ evaluate to the same truth value given the same value of argument $x$. Please note that it would not cause circular proof assuming the program is of finite length. Using the definition of $\sigma_\phi(h)$ in Figure~\ref{fig:axiom-aux}, we can show $\sigma_{\attnum(\phi)}(h) = \emph{filter}(h, \lambda x. \attnum(\phi))$. Then $\attnum(R) = \emph{filter}(\attnum(R_1), \lambda x. \attnum(p)) = \emph{filter}(\attnum(Q_1), \lambda x. \attnum(\phi)) = \attnum(Q)$. Thus, $\attnum(\Delta) \uplus \sigma \models \attnum(Q) = \attnum(R)$.
    \item Inductive case: $Q = Q_1 \times Q_2$. Suppose that $\denot{Q_1}_{\sigma, \Delta} = R_1$, $\attnum(\Delta) \uplus \sigma \models \attnum(Q_1) = \attnum(R_1)$, $\denot{Q_2}_{\sigma, \Delta} = R_2$, $\attnum(\Delta) \uplus \sigma \models \attnum(Q_2) = \attnum(R_2)$, and $\denot{Q_1 \times Q_2}_{\sigma, \Delta} = R$, where
    \[
        R = \emph{foldl}((\lambda ys. \lambda y. \emph{append}(ys, \emph{map}(R_2, \lambda z. \emph{merge}(y, z)))), ~ \nil, R_1)
    \]
    Observe the auxiliary $\texttt{cat}$ function defined in Figure~\ref{fig:axiom} is essentially list append, so $\texttt{cat}(h_1, h_2) = \emph{append}(h_1, h_2)$. Also observe that $h_1 \times' t_2$ would do $\texttt{cat}(h_1, h_2)$ for each $h_2$ in list $t_2$. Thus, $h_1 \times' t_2 = \emph{map}(t_2, \lambda z. \emph{append}(h_1, z))$. The overall $\times$ operator in Figure~\ref{fig:axiom} performs a left hold with operator $\times'$. Therefore,
    \[
        \attnum(Q_1) \times \attnum(Q_2) = \emph{foldl}((\lambda ys. \lambda y. \emph{append}(ys, \emph{map}(\attnum(Q_2), \lambda z. \emph{append}(y, z)))), ~ \nil, \attnum(Q_1))
    \]
    Combined with the correctness of $\attnum$, we know
    \[
    \begin{array}{r c l}
        \attnum(R) &=& \emph{foldl}((\lambda ys. \lambda y. \emph{append}(ys, \emph{map}(\attnum(R_2), \lambda z. \emph{append}(y, z)))), ~ \nil, \attnum(R_1)) \\
                   &=& \emph{foldl}((\lambda ys. \lambda y. \emph{append}(ys, \emph{map}(\attnum(Q_2), \lambda z. \emph{append}(y, z)))), ~ \nil, \attnum(Q_1)) \\
                   &=& \attnum(Q)
    \end{array}
    \]
    Thus, $\attnum(\Delta) \uplus \sigma \models \attnum(Q) = \attnum(R)$.
    \item Inductive case: $Q = Q_1 \cup Q_2$. Suppose that $\denot{Q_1}_{\sigma, \Delta} = R_1$, $\attnum(\Delta) \uplus \sigma \models \attnum(Q_1) = \attnum(R_1)$, $\denot{Q_2}_{\sigma, \Delta} = R_2$, $\attnum(\Delta) \uplus \sigma \models \attnum(Q_2) = \attnum(R_2)$, and $\denot{Q_1 \cup Q_2}_{\sigma, \Delta} = R$, where $R = \emph{append}(R_1, R_2)$. Note that the $\cup$ operator defined in Figure~\ref{fig:axiom} is essentially list append, so we have $\attnum(Q_1) \cup \attnum(Q_2) = \emph{append}(\attnum(Q_1), \attnum(Q_2))$. Then $\attnum(R) = \emph{append}(\attnum(R_1), \attnum(R_2)) = \emph{append}(\attnum(Q_1), \attnum(Q_2)) = \attnum(Q)$. Thus, $\attnum(\Delta) \uplus \sigma \models \attnum(Q) = \attnum(R)$.
    \item Inductive case: $Q = Q_1 - Q_2$. Suppose that $\denot{Q_1}_{\sigma, \Delta} = R_1$, $\attnum(\Delta) \uplus \sigma \models \attnum(Q_1) = \attnum(R_1)$, $\denot{Q_2}_{\sigma, \Delta} = R_2$, $\attnum(\Delta) \uplus \sigma \models \attnum(Q_2) = \attnum(R_2)$, and $\denot{Q_1 - Q_2}_{\sigma, \Delta} = R$, where
    \[
        R = \emph{foldl}((\lambda ys. \lambda y. \emph{delete}(y, ys)), R_1, R_2)
    \]
    Observe the auxiliary $-'$ operator defined in Figure~\ref{fig:axiom} is essentially list delete, and the minus operator $-$ performs a left fold using $-'$, so we have
    \[
        \attnum(Q_1) - \attnum(Q_2) = \emph{foldl}((\lambda ys. \lambda y. \emph{delete}(y, ys)), \attnum(Q_1), \attnum(Q_2))
    \]
    Then we know,
    \[
    \begin{array}{r c l}
        \attnum(R) &=& \emph{foldl}((\lambda ys. \lambda y. \emph{delete}(y, ys)), \attnum(R_1), \attnum(R_2)) \\
                   &=& \emph{foldl}((\lambda ys. \lambda y. \emph{delete}(y, ys)), \attnum(Q_1), \attnum(Q_2)) \\
                   &=& \attnum(Q)
    \end{array}
    \]
    Thus, $\attnum(\Delta) \uplus \sigma \models \attnum(Q) = \attnum(R)$.
\end{enumerate}
\end{proof}

%%%%%%%%%%%%%%%%%%%%%%%%%%%%%%%%%%%%%%%%%%%%%%%%%%%%%%%%%%%%%%%%
% sp insertion
%%%%%%%%%%%%%%%%%%%%%%%%%%%%%%%%%%%%%%%%%%%%%%%%%%%%%%%%%%%%%%%%

\begin{lemma}(\textbf{$sp$ for insertion})\label{thm:sp-ins}
Suppose $U = \ins(R, \{a_1 : v_1, \ldots , a_n : v_n\})$, and let $\Phi$ be a $\tra$ formula. If $(\Delta, \sigma) \sim \Phi$ and $\denot{U}_{\sigma, \Delta} = \Delta'$, then $(\Delta', \sigma) \sim sp(\Phi, U)$.
\end{lemma}

\begin{proof}
Assume $t' = \{a_1 : v_1, \ldots, a_n : v_n\}$ and $t = [v_1, \ldots, v_n]$, then by the definition of $sp$ we know $sp(\Phi, \ins(R, t')) = \exists z.~(R = z \cup [t] \land \Phi[z/R])$. The problem becomes given
\[
\begin{array}{r l l}
    (1) & \attnum(\Delta) \uplus \sigma \models \Phi & \text{(by definition of $(\Delta, \sigma) \sim \Phi$)} \\
    (2) & \attnum(\Delta) \uplus \sigma \models R = c & \text{(assuming $R$ evaluates to $c$ under $\Delta, \sigma$)} \\
    (3) & \sigma \models t = c' & \text{(assuming $t$ evaluates to $c'$ under $\sigma$)} \\
\end{array}
\]
we would like to prove $\attnum(\Delta') \uplus \sigma \models \exists z. (R = z \cup [t] \land \Phi[z/R])$, where $\Delta' = \Delta[R \leftarrow \emph{append}(R, t[\sigma])]$.

In fact, $\attnum(\Delta) \uplus \sigma \uplus [z \mapsto c] \models \Phi[z/R]$ holds directly from (1) and (2). Since $R$ does not occur in $\Phi[z/R]$ any more, it can be assigned to any value. Therefore,
\[
    \attnum(\Delta') \uplus \sigma \uplus [z \mapsto c] \models \Phi[z/R] \qquad (*)
\]

Also, $\Delta'(R) = \emph{append}(c, c')$ holds because of (2), (3) and the fact that $\Delta' = \Delta[R \leftarrow \emph{append}(R, t[\sigma])]$. Since the definition of $\cup$ is the same as list append, we have
\[
    \attnum(\Delta') \uplus \sigma \models R = c \cup c'
\]
Thus, $\attnum(\Delta') \uplus \sigma \uplus [z \mapsto c] \models R = z \cup c'$. Again, since $\sigma \models t=c'$, it holds that
\[
    \attnum(\Delta') \uplus \sigma \uplus [z \mapsto c] \models R = z \cup t \qquad (**)
\]
Put $(*)$ and $(**)$ together, we have
\[
    \attnum(\Delta') \uplus \sigma \uplus [z \mapsto c] \models R = z \cup t \land \Phi[z/R]
\]
Hence, 
\[
    \attnum(\Delta') \uplus \sigma \models \exists z. ~ (R = z \cup t \land \Phi[z/R])
\]
i.e., $(\Delta', \sigma) \sim sp(\Phi, U)$.
\end{proof}

%%%%%%%%%%%%%%%%%%%%%%%%%%%%%%%%%%%%%%%%%%%%%%%%%%%%%%%%%%%%%%%%
% sp deletion
%%%%%%%%%%%%%%%%%%%%%%%%%%%%%%%%%%%%%%%%%%%%%%%%%%%%%%%%%%%%%%%%

\begin{lemma}(\textbf{$sp$ for deletion})\label{thm:sp-del}
Suppose $U = \del(R, \phi)$, and let $\Phi$ be a $\tra$ formula. If $(\Delta, \sigma) \sim \Phi$ and $\denot{U}_{\sigma, \Delta} = \Delta'$, then $(\Delta', \sigma) \sim sp(\Phi, U)$.
\end{lemma}

\begin{proof}
By the definition of $sp$ we know $sp(\Phi, \del(R, \phi)) = \exists z.~(R = \sigma_{\neg \attnum(\phi)}(R) \land \Phi[z/R])$. The problem becomes given
\[
\begin{array}{r l l}
    (1) & \attnum(\Delta) \uplus \sigma \models \Phi & \text{(by definition of $(\Delta, \sigma) \sim \Phi$)} \\
    (2) & \attnum(\Delta) \uplus \sigma \models R = c & \text{(assuming $R$ evaluates to $c$ under $\Delta, \sigma$)} \\
    (3) & \denot{\phi}_{\sigma, \Delta, x} = p & \text{(assuming $\phi$ evaluates to $p$ under $\Delta, \sigma$)} \\
\end{array}
\]
we would like to prove $\attnum(\Delta') \uplus \sigma \models \exists z. (R = \sigma_{\attnum(\neg \phi)}(R) \land \Phi[z/R])$, where
\[
    \Delta' = \Delta[R \leftarrow \emph{filter}(\Delta(R), \lambda x. \neg \denot{\phi}_{\sigma, \Delta, x})]
\]

Observe that $\attnum(\Delta) \uplus \sigma \uplus [z \mapsto c] \models \Phi[z/R]$ holds directly from (1) and (2). Since $R$ does not occur in $\Phi[z/R]$, it can be assigned to any value. Therefore,
\[
    \attnum(\Delta') \uplus \sigma \uplus [z \mapsto c] \models \Phi[z/R] \qquad (*)
\]

Also observe that $\Delta'(R) = \emph{filter}(c, \lambda x. \neg p)$ holds because of (2) and (3). From Lemma~\ref{thm:sp-pred}, we know $p$ and $\phi$ evaluate to same truth value for the same value of argument $x$, and using definition of $\sigma_\phi(h)$ in Figure~\ref{fig:axiom-aux}, we can show $\sigma_{\attnum(\neg \phi)}(h)=\emph{filter}(h, \lambda x. \attnum(\neg \phi))$. Thus,
\[
    \attnum(\Delta') \uplus \sigma \models R = \sigma_{\attnum(\neg \phi)}(c)
\]
and
\[
    \attnum(\Delta') \uplus \sigma \uplus [z \mapsto c] \models R = \sigma_{\attnum(\neg \phi)}(z) \qquad (**)
\]
Combining $(*)$ and $(**)$, we have
\[
    \attnum(\Delta') \uplus \sigma \uplus [z \mapsto c] \models R = \sigma_{\attnum(\neg \phi)}(z) \land \Phi[z/R]
\]
Hence, 
\[
    \attnum(\Delta') \uplus \sigma \models \exists z. ~ (R = \sigma_{\attnum(\neg \phi)}(z) \land \Phi[z/R])
\]
i.e., $(\Delta', \sigma) \sim sp(\Phi, U)$.
\end{proof}

%%%%%%%%%%%%%%%%%%%%%%%%%%%%%%%%%%%%%%%%%%%%%%%%%%%%%%%%%%%%%%%%
% sp update
%%%%%%%%%%%%%%%%%%%%%%%%%%%%%%%%%%%%%%%%%%%%%%%%%%%%%%%%%%%%%%%%

\begin{lemma}(\textbf{$sp$ for update})\label{thm:sp-upd}
Suppose $U = \upd(R, \phi, a, v)$, and let $\Phi$ be a $\tra$ formula. If $(\Delta, \sigma) \sim \Phi$ and $\denot{U}_{\sigma, \Delta} = \Delta'$, then $(\Delta', \sigma) \sim sp(\Phi, U)$.
\end{lemma}

\begin{proof}
    By the definition of $sp$ we know
\[
    sp(\Phi, \upd(R, \phi, a, v)) = \exists z.~((R = \sigma_{\neg \attnum(\phi)}(R) \cup \sigma_{\attnum(\phi)}(R)\awr{\attnum(a)}{v}) \land \Phi[z/R])
\]
The problem becomes given
\[
\begin{array}{r l l}
    (1) & \attnum(\Delta) \uplus \sigma \models \Phi & \text{(by definition of $(\Delta, \sigma) \sim \Phi$)} \\
    (2) & \attnum(\Delta) \uplus \sigma \models R = c & \text{(assuming $R$ evaluates to $c$ under $\Delta, \sigma$)} \\
    (3) & \denot{\phi}_{\sigma, \Delta, x} = p & \text{(assuming $\phi$ evaluates to $p$ under $\Delta, \sigma$)} \\
\end{array}
\]
we would like to prove $\attnum(\Delta') \uplus \sigma \models \exists z. ((R = \sigma_{\attnum(\neg \phi)}(R) \cup \sigma_{\attnum(\phi)}(R)\awr{\attnum(a)}{v}) \land \Phi[z/R])$, where
\[
    \Delta' = \Delta
    \left [
        R \leftarrow \emph{append} \left (
            \begin{array}{l}
            \emph{filter}(\Delta(R), \lambda x. {\neg \denot{\phi}}_{\sigma, \Delta, x}), \\
            \emph{map}(\emph{filter}(\Delta(R), \lambda x.  {\denot{\phi}}_{\sigma, \Delta, x}), \lambda y. y[a \leftarrow v[\sigma]])
            \end{array}
        \right )
    \right ]
\]

Observe that $\attnum(\Delta) \uplus \sigma \uplus [z \mapsto c] \models \Phi[z/R]$ holds directly from (1) and (2). Since $R$ does not occur in $\Phi[z/R]$, it can be assigned to any value. Therefore,
\[
    \attnum(\Delta') \uplus \sigma \uplus [z \mapsto c] \models \Phi[z/R] \qquad (*)
\]

Please note that given premises (2) and (3), $\Delta'(R) = \emph{append}(R_1, \emph{map}(R_2, \lambda y. y[i \leftarrow v[\sigma]]))$, where $R_1 = \emph{filter}(c, \lambda x. \neg p)$ and $R_2 = \emph{filter}(c, \lambda x. p)$. From Lemma~\ref{thm:sp-pred}, we know $p$ and $\phi$ evaluate to same truth value for the same value of argument $x$, and using definition of $\sigma_\phi(h)$ in Figure~\ref{fig:axiom-aux}, we can show $\sigma_{\attnum(\phi)}(h) = \emph{filter}(h, \lambda x. \attnum(\phi))$ and $\sigma_{\attnum(\neg \phi)}(h) = \emph{filter}(h, \lambda x. \attnum(\neg \phi))$. Thus,
\[
\attnum(\Delta') \uplus \sigma \uplus [z \mapsto c] \models R_1 =\sigma_{\attnum(\neg \phi)}(c) \qquad (**)
\]
\[
\attnum(\Delta') \uplus \sigma \uplus [z \mapsto c] \models R_2 =\sigma_{\attnum(\phi)}(c) \qquad 
\]

Also, from the definition of $\cdot \awr{\cdot}{\cdot}$, we have $L \awr{\attnum(a)}{v} = \emph{map}(L, \lambda y. ~ y[\attnum(a) \leftarrow v])$. Let $R_3 = \emph{map}(R_2, \lambda y. y[a \leftarrow v[\sigma]])$ for convenience, then
\[
    \attnum(\Delta') \uplus \sigma \uplus [z \mapsto c] \models R_3 = \sigma_{\attnum(\phi)}(c)\awr{\attnum(a)}{v} \qquad (***)
\]
Combining $(**)$, $(***)$, $\Delta'(R) = \emph{append}(R_1, R_3)$, and the definition of $\cup$ is the same as list append, we get
\[
    \attnum(\Delta') \uplus \sigma \uplus [z \mapsto c] \models R = \sigma_{\attnum(\neg \phi)}(c) \cup \sigma_{\attnum(\phi)}(c)\awr{\attnum(a)}{v}
\]
and thus,
\[
    \attnum(\Delta') \uplus \sigma \uplus [z \mapsto c] \models R = \sigma_{\attnum(\neg \phi)}(z) \cup \sigma_{\attnum(\phi)}(z)\awr{\attnum(a)}{v}
\]
Put together with $(*)$,
\[
    \attnum(\Delta') \uplus \sigma \uplus [z \mapsto c] \models (R = \sigma_{\attnum(\neg \phi)}(z) \cup \sigma_{\attnum(\phi)}(z)\awr{\attnum(a)}{v}) \land \Phi[z/R]
\]
Hence,
\[
    \attnum(\Delta') \uplus \sigma \models \exists z. ~ ((R = \sigma_{\attnum(\neg \phi)}(z) \cup \sigma_{\attnum(\phi)}(z)\awr{\attnum(a)}{v}) \land \Phi[z/R])
\]
i.e., $(\Delta', \sigma) \sim sp(\Phi, U)$.
\end{proof}

%%%%%%%%%%%%%%%%%%%%%%%%%%%%%%%%%%%%%%%%%%%%%%%%%%%%%%%%%%%%%%%%
% sp soundness
%%%%%%%%%%%%%%%%%%%%%%%%%%%%%%%%%%%%%%%%%%%%%%%%%%%%%%%%%%%%%%%%

\vspace{0.1in}

\begin{proof}[Proof of Theorem~\ref{thm:sp}] \textbf{(Soundness of sp)}
By structural induction on $U$.
\begin{enumerate}
    \item Base case: $U = \ins(R, \{a_1 : v_1, \ldots, a_n : v_n\})$. Proven by Lemma~\ref{thm:sp-ins}.
    \item Base case: $U = \del(R, \phi)$. Proven by Lemma~\ref{thm:sp-del}.
    \item Base case: $U = \upd(R, \phi, a, v)$. Proven by Lemma~\ref{thm:sp-upd}.
    \item Inductive case: $U = U_1 ; U_2$. Suppose for any given $\tra$-formula $\Phi_1$ such that if $(\Delta_1, \sigma) \sim \Phi_1$ holds, $\Delta'_1 = \denot{U_1}_{\sigma, \Delta_1}$, then $(\Delta'_1, \sigma) \sim sp(\Phi_1, U_1)$. Also suppose for any given $\tra$-formula $\Phi_2$ such that if $(\Delta_2, \sigma) \sim \Phi_2$ holds, $\Delta'_2 = \denot{U_2}_{\sigma, \Delta_2}$, then $(\Delta'_2, \sigma) \sim sp(\Phi_2, U_2)$. Now consider $U = U_1; U_2$, $\tra$-formula $\Phi$, and $\Delta, \sigma$ such that $(\Delta, \sigma) \sim \Phi$. Let $\Delta_1 = \Delta$, $\Phi_1 = \Phi$ and $\Delta'' = \denot{U_1}_{\sigma, \Delta}$, we know $(\Delta'', \sigma) \sim sp(\Phi, U_1)$. Then let $\Delta_2 = \Delta''$ and $\Phi_2 = sp(\Phi, U_1)$, and $\Delta' = \denot{U_2}_{\sigma, \Delta''}$, we know $(\Delta', \sigma) \sim sp(sp(\Phi, U_1), U_2)$. Also note that $\Delta' = \denot{U_2}_{\sigma, \Delta''} = \denot{U_1 ; U_2}_{\sigma, \Delta}$, we have proven the agreement.
\end{enumerate}
\end{proof}